\documentclass[imslayout,preprint]{imsart}
\usepackage[left=1.25in,right=1.25in,top=1in,bottom=1in]{geometry}
\usepackage[T1]{fontenc}
\usepackage{ae,aecompl}

\usepackage{amsthm,amssymb,amsmath}
\usepackage{natbib}
\usepackage{graphicx}
\usepackage{enumerate}
\usepackage{algorithm}
\usepackage{fancyvrb}

\newcommand{\fracfloat}{.87} 

\startlocaldefs%
\newtheorem{theorem}{Theorem}
\newtheorem{lemma}[theorem]{Lemma}

\theoremstyle{definition}

\theoremstyle{remark}

\newtheorem{example}{Example}

\def\argmin{\mathop{\rm arg\,min}\limits}%
\def\argmax{\mathop{\rm arg\,max}\limits}%
\let\hat\widehat
\let\tilde\widetilde
\def\given{{\,|\,}}
\def\Pr{{\ensuremath{\mathbb P}}}%
\def\Exp{{\ensuremath{\mathbb E}}}%
\endlocaldefs%

\begin{document}
\begin{frontmatter}%
\title{Bayesian Centroid Estimation for Motif Discovery}
\runtitle{Bayesian Centroid Estimation for Motif Discovery}%

\begin{aug}
\author{\fnms{Luis} \snm{Carvalho}%
\thanksref{t1}\ead[label=e1]{lecarval@math.bu.edu}}
\thankstext{t1}{Supported by NSF grant DMS-1107067.}
\runauthor{L. Carvalho}
\affiliation{Boston University}
\address{Department of Mathematics and Statistics\\
Boston University\\
Boston, Massachusetts 02215\\
\printead{e1}}
\end{aug}

\begin{abstract}
Biological sequences may contain patterns that are signal important
biomolecular functions; a classical example is regulation of gene expression
by transcription factors that bind to specific patterns in genomic promoter
regions. In motif discovery we are given a set of sequences that share a
common motif and aim to identify not only the motif composition, but also the
binding sites in each sequence of the set. We present a Bayesian model that is
an extended version of the model adopted by the Gibbs motif sampler, and
propose a new centroid estimator that arises from a refined and meaningful
loss function for binding site inference. We discuss the main advantages of
centroid estimation for motif discovery, including computational convenience,
and how its principled derivation offers further insights about the posterior
distribution of binding site configurations. We also illustrate, using
simulated and real datasets, that the centroid estimator can differ from the
maximum a posteriori estimator.
\end{abstract}



\end{frontmatter}%

\section{Introduction}
In motif discovery we are given a set of sequences that share a common motif
and aim to identify the motif composition---the frequency of symbols for each
position in the pattern---and the positions in each sequence where the motifs
are. It is assumed that the motifs are significantly different, in
composition, from sequence background. This problem has gained attention and
relevance in the past 25 years mainly due to biological applications; a
classical example is regulation of gene expression by transcription factors
that bind to specific motifs in genomic promoter regions
\citep{macisaac06,guhathakurta06,sandve06}.
For this reason, we refer to the positions where the motifs are realized in
the sequences as ``binding sites''.

Due to its importance, hundreds of procedures have been proposed for motif
discovery \citep{hu05,tompa05}. While some approaches seek to characterize
motifs and their binding sites using dictionary methods that capture
over-representation of words as evidence \citep{regnier04,pavesi04}, it is
common to represent motif compositions by a position weight matrix
\citep{stormo00} and specify a parametric model where sequences are generated
conditionally on motif and background compositions and binding sites. Binding
sites can then be regarded as missing data; parameters for the compositions
can be estimated using expectation-maximization \citep{dempster77} in a
frequentist setup, as in MEME \citep{bailey95}, or assigned a prior
distribution in a Bayesian setup \citep{lawrence93,neuwald95}.

Following the Bayesian model from \citep{liu95}, we assume that there is only
one motif of \emph{fixed} length $L$ and that sequences are generated
conditionally independently according to a product multinomial model given
binding site positions and motif and background compositions. Thus, for an
alphabet $\mathcal{S}$, we define $\theta_0 = (\theta_{0,s})_{s \in
\mathcal{S}}$ as background probabilities of generating each letter in
$\mathcal{S}$ and, for each position $i = 1, \ldots, L$ in the motif,
$\theta_i = (\theta_{i,s})_{s \in \mathcal{S}}$ as the probabilities of
generating each letter at the $i$-th position in the motif. To simplify the
notation we denote $\Theta = (\theta_0, \theta_1, \ldots, \theta_L)$. As in
\citep{liu95}, we set a conjugate Dirichlet prior for $\Theta$.

Product multinomial and product Dirichlet models are justified as a good
working, first approximation based on position independence. There are many
extensions to this model that consider DNA strand complementarity
\citep{roth98}, a more informative Markov structure for the background
composition \citep{liu01}, and an explicit representation of the number of
binding sites per sequence \citep{thijs02}. However, since we will be
discussing a new inferential procedure, we adopt an extended model that yields
a feasible computational method while still retaining a realistic
interpretation and allows us to focus the discussion on the proposed
estimator.

Motif discovery is considered a hard problem since motifs are usually short
relative to sequence length and have a composition that might be hard to
distinguish from background (see, for instance, \citep{hu05}.) It is then
imperative to rely on more refined, informative estimation methods that better
glean information from the posterior distribution of binding site
configurations. Discrete inferential methods with this goal have recently been
proposed, including the median probability model of \citet{barbieri04} and the
centroid estimator \citep{ding05,carvalho08}. Centroid estimation, in
particular, has been successfully used for motif discovery \citep{thompson07},
including models that account for sequence conservation \citep{newberg07}.

In this paper we present a Bayesian model for motif discovery on multiple
sequences with multiple possible binding sites and formalize a new flavor of
inference based on centroid estimation. As we will argue, the proposed
estimator offers a good representative of the posterior space of binding site
configurations; moreover, as a by-product of its derivation, we obtain
informative summaries of the distribution of posterior mass. We start the
discussion by addressing a simple case when there is only one sequence and we
accept only one binding site; next we extend the presentation to include
multiple binding sites; then, we treat the full case when $\Theta$ is random,
in a fully Bayesian approach. Finally, we offer some concluding remarks and
directions for future work in the last section.

\section{One sequence, one binding site}
Suppose we observe a sequence $R$, $|R| \doteq n$, and wish to infer the
location of the only binding site $Y$, $Y \in \{1, \ldots, n-L+1\}$. Setting a
non-informative prior on $Y$, $\Pr(Y) = (n - L + 1)^{-1}$, we have the
posterior:
\[
\Pr(Y \given R, \Theta) =
\frac{\Pr(R \given Y, \Theta) \Pr(Y \given \Theta)}
{\sum_{\tilde{Y}=1}^{n-L+1} \Pr(R \given \tilde{Y}, \Theta) \Pr(\tilde{Y}
\given \Theta)}
= \frac{\Pr(R \given Y, \Theta)}
{\sum_{\tilde{Y}=1}^{n-L+1} \Pr(R \given \tilde{Y}, \Theta)}.
\]
The likelihood, as previously stated, follows a product multinomial
distribution given $Y$:
\[
\Pr(R \given Y, \Theta) = \prod_{s \in \mathcal{S}}
\prod_{j \in BG} \theta_{0,s}^{I(R_j = s)}
\prod_{j=1}^L \theta_{j,s}^{I(R_{Y-j+1} = s)},
\]
where $j \in BG$ means position $j$ in background.

One traditional estimator is the MAP estimator,
\[
\hat{Y}_M = \argmax_{\tilde{Y}=1,\ldots,n-L+1}
\Pr(\tilde{Y} \given R, \Theta),
\]
but we argue for an estimator that accounts for differences in positions when
comparing binding site configurations. Using Bayesian decision theory
\citep{berger85} we look for an estimator that minimizes, on average, a more
refined loss function $H$:
\begin{equation}
\label{eq:estimator}
\hat{Y}_C = \argmin_{\tilde{Y}=1,\ldots,n-L+1}
\Exp_{Y \given R, \Theta} \big[ H(\tilde{Y}, Y) \big].
\end{equation}
We adopt a generalized Hamming loss $H$,
\[
H(\tilde{Y}, Y) = \sum_{i=1}^n h(l_i(\tilde{Y}), l_i(Y)),
\]
where $l_i(Y)$ returns the ``state'' of position $i$: if $i$ is a background
position, $l_i(Y) = 0$, otherwise $l_i(Y) = Y - i + 1$, that is, $l_i(Y)$
returns the position in the motif. Loss function $H$ compares configurations
position-wise according to $h$, which in turn compares states.
One option for $h$ when $\Theta$ is known is a probability distance, the
symmetric Kullback-Leibler distance,
\[
h(i,j) = D_{KL}(\theta_i \,||\, \theta_j) + D_{KL}(\theta_j \,||\, \theta_i)
= \sum_{s \in \mathcal{S}}
\theta_{i,s} \log \frac{\theta_{i,s}}{\theta_{j,s}}
+ \theta_{j,s} \log \frac{\theta_{j,s}}{\theta_{i,s}}, 
\]
for $i,j = 0, 1, \ldots, L$.

It is, however, not common to have such an informed loss function. An
alternative metric arises by simply allowing $\theta_{j,s} \doteq \theta_s \ne
\theta_{0,s}$ for all $s \in \mathcal{S}$ and $j = 1, \ldots, L$ in the motif.
In this case, if $m(i) \doteq I(i > 0)$ indicates if state $i$ is a motif
state,
\[
h(i,j) = h(m(i), m(j)) = I(m(i) \ne m(j)) \Bigg[
\sum_{s \in \mathcal{S}} 
\theta_s \log \frac{\theta_s}{\theta_{0,s}}
+ \theta_{0,s} \log \frac{\theta_{0,s}}{\theta_s} \Bigg]. 
\]
Since we are ultimately concerned with the argument of a minimum, as per
Equation~\ref{eq:estimator}, we can define the loss function up to a shift and
(positive) scale. Thus, for our inferential purposes it suffices to define
$h(i, j) = I(m(i) \ne m(j))$ to obtain a loss $H$ that accounts for overlap in
binding sites. Such metric is commonly adopted to measure binding site level
accuracy, as in the performance coefficients
in~\citep{pevzner00,hu05,tompa05}. From now on we will be focusing on this
minimally informed loss function.

Estimator $\hat{Y}_C$ is a \emph{generalized centroid estimator}; for
instance, if $h$ is a common zero-one loss, $h(i,j) = I(i \ne j)$, $H$
corresponds to Hamming loss, and thus $\hat{Y}_C$ is the regular centroid
estimator~\citep{ding05,carvalho08}. As \citet{carvalho08} argue, centroid
estimators more effectively represent the space since they are closer to
posterior means; in contrast, it can be shown that $\hat{Y}_M$ arises from a
zero-one loss function which yields the posterior mode \citep{besag86}.

Let us now derive more specific expressions for $H$ and $\hat{Y}_C$. We first
notice that if $|\tilde{Y}-Y| \ge L$ then the binding sites do not overlap and
so $H(\tilde{Y}, Y) = 2 \sum_{j=1}^L h(j,0) \doteq H^*$, the null overlap
distance between two configurations. Alternatively, when $|\tilde{Y}-Y| < L$
then
\begin{equation}
\label{eq:loss}
H(\tilde{Y}, Y) = 
\sum_{j=1}^{|\tilde{Y}-Y|} h(j, 0)
+ \sum_{j=L - |\tilde{Y}-Y| + 1}^L h(j, 0)
+ \sum_{j=1}^{L-|\tilde{Y}-Y|} h(j, j+|\tilde{Y}-Y|),
\end{equation}
since the common backgrounds in $\tilde{Y}$ and $Y$ do not affect
$H(\tilde{Y}, Y)$, the first two terms above account for the left and right
``tails'' where binding sites in one sequence are matched with background in
the other sequence, and the last term accounts for the overlap in binding
sites. We also note that $H(\tilde{Y}, Y)$ is actually a function of
$|\tilde{Y}-Y|$.

Instead of a loss function we can also define our estimator in terms of a
\emph{gain} function
$G(\tilde{Y}, Y)
\doteq 1 - H(\tilde{Y}, Y) / H^*$. Note that $0 \le
G(\tilde{Y}, Y) \le 1$; in particular, when $|\tilde{Y}-Y| \ge L$ there is no
gain, $G(\tilde{Y}, Y) = 0$, and if $\tilde{Y}=Y$ we have $G(\tilde{Y}, Y) =
1$. As a consequence, we can simply write $G(\tilde{Y}, Y) =
I(|\tilde{Y}-Y|<L) (1 - H(\tilde{Y}, Y)/H^*)$ with
$H$ from Equation~\ref{eq:loss}. Noting that $G$, like $H$, is also a function
of $|\tilde{Y}-Y|$, we obtain the following characterization:

\begin{theorem}
\label{thm:cent1}
The centroid estimator $\hat{Y}_C$ is
\[
\hat{Y}_C = \argmax_{\tilde{Y} = 1, \ldots, n - L + 1}
G(\tilde{Y}, \cdot) * \Pr(\cdot \given R, \Theta),
\]
a \emph{convolution} between $G$ and the posterior distribution on $Y$.
\end{theorem}

\begin{proof}
The result follows directly from the definition in
Equation~\ref{eq:estimator}:
\[
\begin{split}
\hat{Y}_C & = \argmin_{\tilde{Y} = 1, \ldots, n - L + 1}
\Exp_{Y \given R, \Theta} \big[ H(\tilde{Y}, Y) \big] \\
& = \argmax_{\tilde{Y} = 1, \ldots, n - L + 1}
\Exp_{Y \given R, \Theta} \big[
I(|\tilde{Y}-Y|<L) (1 - H(\tilde{Y}, Y)/H^*)
\big] \\
& = \argmax_{\tilde{Y} = 1, \ldots, n - L + 1}
\sum_{Y = \max\{1, \tilde{Y}-L+1\}}^{\min\{n-L+1, \tilde{Y}+L-1\}}
G(\tilde{Y}, Y) \Pr(Y \given R, \Theta) \\
& = \argmax_{\tilde{Y} = 1, \ldots, n - L + 1}
G(\tilde{Y}, \cdot) * \Pr(\cdot \given R, \Theta),
\end{split}
\]
as required.
\end{proof}

When contrasted to $\hat{Y}_M$ we can see the effect of having a higher
resolution loss function: $\hat{Y}_C$ gathers probability support from nearby,
relative to $H$, binding site configurations instead of just picking the most
likely configuration. The following example should give us some insight into
this new estimator.

\begin{example}
\label{ex:sbs}
Consider the following sequence of length $n=200$ from the nucleotide alphabet
$\mathcal{S} = \{$\texttt{A}, \texttt{C}, \texttt{G}, \texttt{T}$\}$,
\medskip
\begin{center}
\begin{BVerbatim}[boxwidth=auto]
        10        20        30        40        50
         |         |         |         |         |
GCCACTTTCGGGCCCGTGTCTAACGCACCACGGGCTACGTGACGGTGTGG
CTCTATACTGACGACGTGAACCAAGCTTTACTGAAGGACTTGCTGTTCCC
CGACCCATTTCCTGCCAGAACCTCTGACCAGTGTCTAGGGCTATCGCCCG
TGATGTCTCATGGCGACGCGCGAGGCGGTTGCTCGCCTCACTCCGTTCTG
\end{BVerbatim}
\end{center}
\medskip
and a motif of length $L=6$ with parameters $\Theta$ given by
Table~\ref{tab:theta}.
\begin{table}[bht]
\caption{Background and motif compositions: background is assumed to be
\texttt{CG}-rich, while the motif represents a canonical palindromic E-box,
\texttt{CACGTG} \citep{murrea89}.}
\label{tab:theta}
\begin{tabular*}{\textwidth}
{c@{\extracolsep{\fill}}c@{\extracolsep{\fill}}c@{\extracolsep{\fill}}c@{\extracolsep{\fill}}c@{\extracolsep{\fill}}c@{\extracolsep{\fill}}c@{\extracolsep{\fill}}c} \hline
$\mathcal{S}$ & $\theta_0$ & $\theta_1$ & $\theta_2$ & $\theta_3$ & $\theta_4$ & $\theta_5$ & $\theta_6$ \\ \hline
\texttt{A} & $0.2$ & $0.1$ & $0.7$ & $0.1$ & $0.1$ & $0.1$ & $0.1$ \\
\texttt{C} & $0.3$ & $0.7$ & $0.1$ & $0.7$ & $0.1$ & $0.1$ & $0.1$ \\
\texttt{G} & $0.3$ & $0.1$ & $0.1$ & $0.1$ & $0.7$ & $0.1$ & $0.7$ \\
\texttt{T} & $0.2$ & $0.1$ & $0.1$ & $0.1$ & $0.1$ & $0.7$ & $0.1$ \\ \hline
\end{tabular*}
\end{table}

Figure~\ref{fig:ex1} shows the conditional marginal posterior $\Pr(Y \given R,
\Theta)$ and the convolution $G * \Pr(\cdot \given R, \Theta)$ used to obtain
the centroid $\hat{Y}_C = 36$, binding at the subsequence \texttt{TACGTG},
close to the consensual motif. Note that since $\Theta$ is very informative
the posterior profile has clear peaks and in this case $\hat{Y}_c =
\hat{Y}_M$, the two estimators coincide.

\begin{figure}[hbt]
\includegraphics[width=\fracfloat\textwidth]{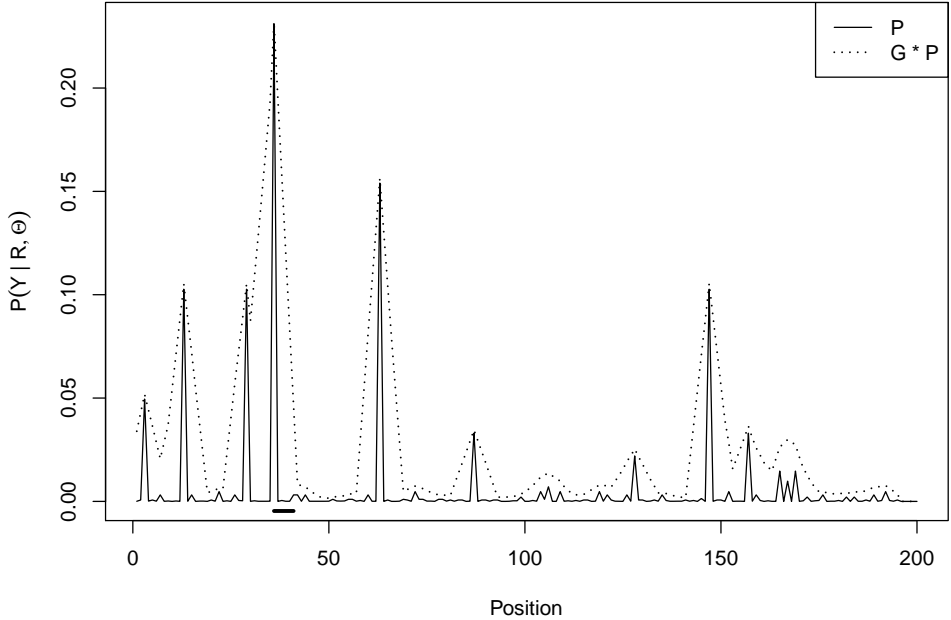}
\caption{Conditional marginal probability distribution $\Pr(Y \given R,
\Theta)$ in solid line and convolution $G * \Pr(\cdot \given R, \Theta)$ in
dotted line. The black thick line close to the axis marks the binding site
corresponding to the centroid $\hat{Y}_C$.}
\label{fig:ex1}
\end{figure}

\end{example}

\section{One sequence, multiple binding sites}
We now allow for multiple binding sites by defining $Y = \{Y_k\}$ as the
collection of binding sites $Y_k$. The likelihood is similar, but accounts for
the multiple binding sites:
\[
\Pr(R \given Y,\Theta) =
\prod_{s \in \mathcal{S}}
\prod_{i \in BG} \theta_{0,s}^{I(R_i=s)}
\prod_{k=1}^{|Y|} \prod_{i=1}^L \theta_{i,s}^{I(R_{Y_k+i-1}=s)}.
\]

Given the ``entropic'' effect of possibly having many binding sites, we need
to adopt a better prior for $Y$ that takes into account the number of possible
configurations for the binding sites. So, instead of naively electing $\Pr(Y)
\propto 1$, we can explore a hierarchical structure: if $c(Y) = |Y|$, the
number of binding sites in $Y$, we note that
$\Pr(Y) = \Pr(Y, c(Y)) = \Pr(Y \given c(Y)) \Pr(c(Y))$ and set
$\Pr(Y \given c(Y)) \propto 1$ and $\Pr(c(Y)) \propto 1$ to obtain
\[
\Pr(Y) = \Pr(Y \given c(Y)) \Pr(c(Y)) = \binom{n - c(Y)(L-1)}{c(Y)}^{-1} \cdot
\frac{1}{C},
\]
where $C \doteq \lfloor n/L \rfloor$ is the maximum number of binding sites
in $R$.

Another, possibly more familiar, approach is to adopt a Markov chain with two
states, background and motif, where the probability of transitioning to
background, either from background or motif, and of starting at background is
$p$. In this case we keep $\Pr(Y \given c(Y))$ as before, but now 
\begin{equation}
\label{eq:priorp}
\Pr(c(Y)) \propto \binom{n - c(Y)(L-1)}{c(Y)} p^{n-c(Y)L} (1-p)^{c(Y)},
\end{equation}
since there needs to be $c(Y)$ transitions to the motif state. This prior
structure offers more flexibility through $p$: we can further set a hyperprior
distribution on $p$, or specify it directly based on the expected number $b$
of binding sites in the sequence; if $n$ is large compared to $b$, as usual,
then $p$ should be close to one, $c(Y)$ is approximately Poisson with mean
$n(1-p)$ and thus $p \doteq 1 - b/n$ becomes a good candidate.

The posterior is then
\[
\begin{split}
\Pr(Y \given  R, \Theta) & = \frac{\Pr(R, Y \given \Theta)}
{\sum_{\tilde{Y}} \Pr(R, \tilde{Y} \given \Theta)} \\
& =
\underbrace{
\frac{\Pr(R, Y \given \Theta)}
{\sum_{\tilde{Y}:c(\tilde{Y})=c(Y)}
\Pr(R, \tilde{Y} \given \Theta)}}_{\Pr(Y \given c(Y), R, \Theta)}
\cdot
\underbrace{
\frac{\sum_{\tilde{Y}:c(\tilde{Y})=c(Y)} \Pr(R, \tilde{Y} \given \Theta)}
{\sum_{c=0}^C \sum_{\tilde{Y}:c(\tilde{Y})=c}
\Pr(R,\tilde{Y} \given \Theta)}}_{\Pr(c(Y) \given R, \Theta)}.
\end{split}
\]
By the structure of our prior it follows that
\begin{equation}
\label{eq:Yprior1}
\begin{split}
\Pr(Y \given c(Y), R, \Theta) & =
\frac{\Pr(R \given Y, \Theta) \Pr(Y)}
{\sum_{\tilde{Y}:c(\tilde{Y})=c(Y)} \Pr(R \given \tilde{Y}, \Theta)
\Pr(\tilde{Y})} \\
& = \frac{\Pr(R \given Y, \Theta)}
{\sum_{\tilde{Y}:c(\tilde{Y})=c(Y)} \Pr(R \given \tilde{Y}, \Theta)},
\end{split}
\end{equation}
and
\begin{equation}
\label{eq:Yprior2}
\begin{split}
\Pr(c(Y) \given R, \Theta) & =
\frac{\sum_{\tilde{Y}:c(\tilde{Y})=c(Y)} \Pr(R \given \tilde{Y}, \Theta)
\Pr(\tilde{Y})}
{\sum_{c=0}^C \sum_{\tilde{Y}:c(\tilde{Y})=c}
\Pr(R \given \tilde{Y}, \Theta) \Pr(\tilde{Y})} \\
& =
\frac{\sum_{\tilde{Y}:c(\tilde{Y})=c(Y)} \Pr(R \given \tilde{Y}, \Theta)
\Pr(\tilde{Y} \given c(\tilde{Y})) \Pr(c(\tilde{Y}))}
{\sum_{c=0}^C \sum_{\tilde{Y}:c(\tilde{Y})=c}
\Pr(R \given \tilde{Y}, \Theta) \Pr(\tilde{Y} \given c(\tilde{Y}))
\Pr(c(\tilde{Y}))}.
\end{split}
\end{equation}

This decomposition suggests a good approach to sampling from $\Pr(Y \given R,
\Theta)$: we first sample $c(Y)$ according to $\Pr(c(Y) \given R, \Theta)$ and
then sample $Y$ given the number of binding sites, according to
$\Pr(Y \given c(Y), R, \Theta)$.

As we will see next, we need to work more to obtain a centroid estimator for
the binding sites: we need to establish a hierarchical inferential structure
by first finding centroids for $c(Y) = 1, \ldots, C$ and then proceed to
estimate a global centroid. To this end we find $\Pr(c(Y) \given R, \Theta)$
and then compute marginal posteriors $\Pr(Y_k \given c(Y), R, \Theta)$.

\subsection{Marginal posterior on $c(Y)$}
\label{ssec:postcy}
From Equations~\ref{eq:Yprior1} and~\ref{eq:Yprior2} we observe that we need
to compute $\sum_{\tilde{Y}:c(\tilde{Y})=c} \Pr(R \given \tilde{Y}, \Theta)$
up to a constant to find both conditional posteriors of $c(Y)$ and $Y$ and
thus the posterior $\Pr(Y \given R, \Theta)$. Let us now denote by $R_{i:j}$
the subsequence of $R$ from positions $i$ to $j$ and by $Y_{i:j}$ the binding
sites in $Y$ between $i$ and $j$---that is, all $Y_k$ such that $i \le Y_k \le
j - L + 1$. If we then define \emph{forward sums}
\begin{equation}
F_{c,j} \doteq
\frac{\sum_{\tilde{Y}_{1:j}: c(\tilde{Y}_{1:j}) = c}
\Pr(R_{1:j} \given \tilde{Y}_{1:j}, \Theta)}
{\prod_{i=1}^j \prod_{s \in \mathcal{S}} \theta_{0,s}^{I(R_i=s)}}
\end{equation}
we have that
$\sum_{\tilde{Y}:c(\tilde{Y})=c} \Pr(R \given \tilde{Y},\Theta) \propto
F_{c,n}$. To further simplify the notation, let us define
\[
\lambda(j; \Theta) = 
\prod_{i=1}^L \prod_{s \in \mathcal{S}}
\Bigg(\frac{\theta_{i,s}}{\theta_{0,s}}\Bigg)^{I(R_{j-1+i}=s)},
\]
the composition ratio between motif and background for a binding site starting
at $j$.

The forward sums $F_{c,j}$ can be computed recursively,
\begin{equation}
\label{eq:fcjrec}
F_{c,j} = F_{c,j-1} + F_{c-1,j-L} \lambda(j-L+1; \Theta),
\end{equation}
by considering two options for the tail of the sequence: either having a
background position---and hence the first summand above---or by having a
binding site on the last $L$ positions---and thus requiring the second
summand.

Thus, we have
\begin{equation}
\label{eq:cY}
\Pr(c(Y) \given R, \Theta) =
\frac{F_{c(Y),n} \binom{n-c(Y)(L-1)}{c(Y)}^{-1} \Pr(c(Y))}
{\sum_{c=0}^C F_{c,n} \binom{n-c(L-1)}{c}^{-1}\Pr(c(Y)=c)},
\end{equation}
which yields a straightforward way to sample the posterior $c(Y)$ conditional
on $\Theta$.

\subsection{Marginal posterior on $Y_k$ given $c(Y)$}
\label{ssec:postykcy}
To compute $\Pr(Y_k \given c(Y), R, \Theta)$ we now need backward sums. We can
define them analogously to the forward sums:
\begin{equation}
B_{c,j} \doteq
\frac{\sum_{\tilde{Y}_{j:n}: c(\tilde{Y}_{j:n}) = c}
\Pr(R_{j:n} \given \tilde{Y}_{j:n}, \Theta)}
{\prod_{i=j}^n \prod_{s \in \mathcal{S}} \theta_{0,s}^{I(R_i=s)}},
\end{equation}
and hence 
$\sum_{\tilde{Y}:c(\tilde{Y})=c} \Pr(R \given \tilde{Y},\Theta) \propto
B_{c,1}$, as expected. Moreover, by a similar argument to the previous
subsection, we also have that the backward sums are recursive:
\begin{equation}
\label{eq:bcjrec}
B_{c,j} = B_{c,j+1} + B_{c-1,j+L} \lambda(j; \Theta).
\end{equation}

Having forward and backward sums enable us to readily compute the marginal
posterior on $Y_k$ conditional on $c(Y)$: since
\[
\begin{split}
\Pr(Y_k \given c(Y)=c, R, \Theta) & =
\sum_{Y_1, \ldots, Y_{k-1}, Y_{k+1}, \ldots, Y_c}
\Pr(Y \given c(Y)=c, R, \Theta) \\
& = \sum_{Y_1, \ldots, Y_{k-1}, Y_{k+1}, \ldots, Y_c}
\frac{\Pr(R \given Y, \Theta)}
{\sum_{\tilde{Y}:c(\tilde{Y})=c} \Pr(R \given \tilde{Y}, \Theta)},
\end{split}
\]
and 
\begin{multline*}
\sum_{Y_1, \ldots, Y_{k-1}, Y_{k+1}, \ldots, Y_c}
\Pr(R \given Y, \Theta) =
\sum_{Y_1, \ldots, Y_{k-1}}
\Pr(R_{1:Y_k-1} \given Y_{1:Y_k-1}, \Theta) \\
\cdot \Pr(R_{Y_k:Y_k+L-1} \given Y_{Y_k:Y_k+L-1}, \Theta)
\cdot \sum_{Y_{k+1}, \ldots, Y_c}
\Pr(R_{Y_k+L:n} \given Y_{Y_k+L:n}, \Theta),
\end{multline*}
and thus
\begin{equation}
\label{eq:Yk}
\Pr(Y_k \given c(Y)=c, R, \Theta) =
\frac{F_{k-1, Y_k-1} \lambda(Y_k; \Theta) B_{c-k, Y_k+L}}
{\sum_{\tilde{Y}_k = (k-1)L}^{n-(c-k+1)L+1}
F_{k-1, \tilde{Y}_k-1} \lambda(\tilde{Y}_k; \Theta) B_{c-k, \tilde{Y}_k+L}}.
\end{equation}
Note that
\begin{multline*}
\frac{\sum_{\tilde{Y}:c(\tilde{Y})=c} \Pr(R \given \tilde{Y}, \Theta)}
{\prod_{i=1}^n \prod_{s \in \mathcal{S}} \theta_{0,s}^{I(R_i = s)}} =
F_{c,n} = B_{c,1} \\
= \sum_{\tilde{Y}_k = (k-1)L}^{n-(c-k+1)L+1}
F_{k-1, \tilde{Y}_k-1} \lambda(\tilde{Y}_k; \Theta) B_{c-k, \tilde{Y}_k+L},
\end{multline*}
for $k = 1, \ldots, c$.

Before discussing posterior inference we summarize the results of this
section in Algorithm~\ref{alg:marg}.
\begin{algorithm}[htpb]
\caption{Computes $\Pr(c(Y) \given R, \Theta)$ and
$\Pr(Y_k \given c(Y), R, \Theta)$ for $k = 1, \ldots, c(Y)$.}
\label{alg:marg}
\begin{enumerate}[Step 1.]
\item \emph{(Initialize)} Set $F_{0,0} = B_{0,n+1} = F_{0,j} = B_{0,j} = 1$
for $j = 1, \ldots, n$; for $c = 1, \ldots, C$, set $F_{c,j} =
0$ when $j < cL$ and $B_{c,j} = 0$ when $j > n - cL + 1$.

\item \emph{(Compute forward sums)} For $c = 1, \ldots, C$ and $j = cL+1,
\ldots, n$ do: set $F_{c,j}$ as in Equation~\ref{eq:fcjrec},
\[
F_{c,j} = F_{c,j-1} + F_{c-1,j-L} \lambda(j-L+1; \Theta)
\]

\item \emph{(Compute $\Pr(c(Y) \given R, \Theta)$)} For $c = 0, \ldots, C$ do:
compute marginal posterior $c(Y)$ as in Equation~\ref{eq:cY},
\[
\Pr(c(Y) = c \given R, \Theta) =
\frac{F_{c,n} \binom{n-c(L-1)}{c}^{-1} \Pr(c(Y)=c)}
{\sum_{\tilde{c}=0}^C F_{\tilde{c},n}
\binom{n-\tilde{c}(L-1)}{\tilde{c}}^{-1} \Pr(c(Y)=\tilde{c})}
\]

\item \emph{(Compute backward sums)} For $c = 1, \ldots, C$ and $j = n-cL, \ldots, 1$
do: set $B_{c,j}$ as in Equation~\ref{eq:bcjrec},
\[
B_{c,j} = B_{c,j+1} + B_{c-1,j+L} \lambda(j; \Theta)
\]

\item \emph{(Compute $\Pr(Y_k \given c(Y), R, \Theta)$)} For $c = 1, \ldots,
C$, $k = 1, \ldots, c$, and $Y_k = (k-1)L+1, \ldots, n-(c-k+1)L+1$ do: compute
marginal posterior $Y_k$ given $c(Y)$ as in Equation~\ref{eq:Yk},
\[
\Pr(Y_k \given c(Y)=c, R, \Theta) =
F_{k-1, Y_k-1} \lambda(Y_k; \Theta) B_{c-k, Y_k+L} / F_{c,n}
\]
\end{enumerate}
\end{algorithm}

\subsection{Posterior Inference}
\label{ssec:postinf}
In contrast to the one binding site case from last section, posterior
inference is more difficult since comparing configurations with different
number of binding sites is not amenable to a systematic approach. Our first
approximation is to consider local estimators for each group of configurations
with a fixed number of binding sites and then appeal to a triangle inequality:
\[
H(Y, \hat{Y}) \le H(Y, \hat{Y}_c) + H(\hat{Y}_c, \hat{Y}),
\]
where $Y$ is a configuration with $c$ binding sites, $\hat{Y}_c$ is the
constrained estimator for all configurations with $c$ binding sites, and
$\hat{Y}$ is the (overall) centroid estimator. Recall that for the centroid
estimator we wish to find $\tilde{Y}$ that minimizes
\[
\Exp_{Y \given R, \Theta} \big[ H(\tilde{Y}, Y) \big] =
\sum_{c = 0}^C \sum_{Y : c(Y) = c}
H(\tilde{Y}, Y) \Pr(Y \given R, \Theta).
\]
Using the triangle inequality for each group we then have
\begin{multline}
\label{eq:Hbound}
\Exp_{Y \given R, \Theta} \big[ H(\tilde{Y}, Y) \big] \le
\sum_{c = 0}^C \sum_{Y : c(Y) = c}
\big[ H(\tilde{Y}, \tilde{Y}_c) + H(\tilde{Y}_c, Y) \big]
\Pr(Y \given R, \Theta) \\
= \sum_{c=0}^C \Bigg[
H(\tilde{Y}, \tilde{Y}_c)
+ \sum_{Y:c(Y)=c} H(\tilde{Y}_c, Y) \Pr(Y|c(Y)=c, R, \Theta)
\Bigg] \Pr(c(Y) = c \given R, \Theta),
\end{multline}
where $\tilde{Y}_c$ is an arbitrary point in $\{Y:c(Y)=c\}$. Our task is now
to find an estimator---let us still call it centroid---that minimizes the
right-hand bound in Equation~\ref{eq:Hbound} above. This goal suggests a
two-step strategy:
\begin{enumerate}[1.]
\item For each number of binding sites, $c = 1, \ldots, C$, find the
\emph{local} centroids
\begin{equation}
\label{eq:lcent}
\hat{Y}_c =
\argmin_{\tilde{Y}:c(\tilde{Y}) = c}
\Exp_{Y \given c(Y)=c, R, \Theta} \big[ H(\tilde{Y}, Y) \big]
\end{equation}
as the $\tilde{Y}_c$ in Equation~\ref{eq:Hbound}.

\item Find the \emph{global} centroid given the local centroids
$\{\hat{Y}_c\}_{c=1}^C$,
\begin{equation}
\label{eq:gcent}
\hat{Y} = \argmin_{\tilde{Y}}
\Exp_{c(Y) \given R, \Theta} \big[ H(\hat{Y}_{c(Y)}, \tilde{Y}) \big].
\end{equation}
\end{enumerate}

We note that this strategy does not guarantee that the bound is minimized; the
main goal here is computational convenience. Let us tackle each step of this
heuristic next.

\subsubsection{Local centroids}
Even when the number of binding sites is fixed, minimizing the conditional
posterior expectation of $H(\tilde{Y}, Y)$ can be challenging: we would still
have to consider for each candidate configuration $\tilde{Y}$ the posterior
probability of configurations with all binding sites to the left of the first
binding site in $\tilde{Y}$, in-between binding sites in $\tilde{Y}$, and so
on. We adopt another approximation and decide to minimize a \emph{paired}
Hamming loss $H_A$ where binding site positions are matched according to their
order:
\[
H_A(\tilde{Y}, Y) = \sum_{k=1}^{c(Y)} H_1(\tilde{Y}_k, Y_k),
\]
where $H_1(\tilde{Y}_k, Y_k)$ is Hamming loss when comparing sequences with only
one binding site at $\tilde{Y}_k$ and $Y_k$, respectively, that is, 
$H_1(\tilde{Y}_k, Y_k) = 2\max\{|\tilde{Y}_k - Y_k|, L\}$.
From the definition we have that $H_A$ upper bounds $H$: $H_A(\tilde{Y}, Y)
\ge H(\tilde{Y}, Y)$. As a bad approximation example, if $\tilde{Y}_k =
Y_{k+1}$ for $k = 1, \ldots, c(Y)-1$ then $H_A(\tilde{Y}, Y) = c(Y) L$, since
each pair of binding sites $\tilde{Y}_k$ and $Y_k$ does not overlap, while
$H(\tilde{Y}, Y) = 2L$ since only $Y_1$ and $\tilde{Y}_{c(Y)}$ are in
disagreement with background.

The next result adapts Theorem~\ref{thm:cent1} to yield the paired local
centroids.
\begin{lemma}
\label{lem:lcent}
If $\Pr_k(\cdot \given c(Y)=c, R, \Theta)$ is the marginal conditional
posterior on $Y_k$ then the paired local centroids are
\[
\hat{Y}_c = \argmax_{\tilde{Y}:c(\tilde{Y})=c}
\sum_{k=1}^c
G(\tilde{Y}_k, \cdot) * \Pr_k(\cdot \given c(Y)=c, R, \Theta)
\]
\end{lemma}

\begin{proof}
In the same spirit of Theorem~\ref{thm:cent1}, we use the conditional
estimator in Equation~\ref{eq:lcent} with the paired loss $H_A$:
\[
\begin{split}
\hat{Y}_c & =
\argmin_{\tilde{Y}:c(\tilde{Y})=c}
\Exp_{Y \given c(Y)=c, R, \Theta} \big[ H_A(\tilde{Y}, Y) \big] \\
& =
\argmin_{\tilde{Y}:c(\tilde{Y})=c}
\sum_{Y:c(Y)=c} \sum_{k=1}^c H_1(\tilde{Y}_k, Y_k)
\Pr(Y \given c(Y)=c, R, \Theta) \\
& =
\argmin_{\tilde{Y}:c(\tilde{Y})=c}
\sum_{k=1}^c \sum_{Y_k=(k-1)L+1}^{n-(c-k+1)L+1} H_1(\tilde{Y}_k, Y_k)
\Pr(Y_k \given c(Y)=c, R, \Theta) \\
& =
\argmax_{\tilde{Y}:c(\tilde{Y})=c}
\sum_{k=1}^c
\sum_{Y_k=\max\{(k-1)L+1, \tilde{Y}_k-L\}}^{\min\{n-(c-k+1)L+1, \tilde{Y}_k+L\}}
G(\tilde{Y}_k, Y_k)
\Pr(Y_k \given c(Y)=c, R, \Theta) \\
& =
\argmax_{\tilde{Y}:c(\tilde{Y})=c}
\sum_{k=1}^c
G(\tilde{Y}_k, \cdot) * \Pr_k(\cdot \given c(Y)=c, R, \Theta),
\end{split}
\]
and the result follows.
\end{proof}

We can spot in Lemma~\ref{lem:lcent} the familiar convolutions, but now with
the marginal posteriors $\Pr(Y_k \given c(Y), R, \Theta)$ and in a more
restricted range. We have a nice characterization, but we still have to
optimize a sum to obtain the local centroids; to this end we explore the same
recursive structure that allowed us to compute forward and backward sums. 
Let us define 
$f(\tilde{Y}_k) \doteq G(\tilde{Y}_k, \cdot) * \Pr_k(\cdot \given c(Y)=c, R,
\Theta)$ as the convolution against the marginal posterior on $Y_k$; then we
should have
\begin{equation}
\label{eq:lcent1}
\max_{\tilde{Y}:c(\tilde{Y})=c} \sum_{k=1}^c f(\tilde{Y}_k)
= \max_{\tilde{Y}_c = (c-1)L+1, \ldots, n-cL+1}
\Bigg[ f(\tilde{Y}_c)
+ \max_{\tilde{Y}_1, \ldots, \tilde{Y}_{c-1}} \sum_{k=1}^{c-1} f(\tilde{Y}_k)
\Bigg].
\end{equation}
This important observation allows us to obtain $\hat{Y}_c$ using the dynamic
programming approach listed in Algorithm~\ref{alg:localcentroid}, as
Theorem~\ref{thm:lcent} formalizes.

\begin{algorithm}[htbp]
\caption{Find $\hat{Y}_c$ using dynamic programming.}
\label{alg:localcentroid}
\raggedright \emph{Construct partial maxima and backtrack pointers:}

\begin{enumerate}[Step 1.]
\item Set $m_1(\tilde{Y}_1) = f(\tilde{Y}_1)$ for
$\tilde{Y}_1 = 1, \ldots, n-cL+1$.

\item For $k = 2, \ldots, c$ and
$\tilde{Y}_k = (k-1)L + 1, \ldots, n-(c-k+1)L + 1$ do:
set backtrack pointers
\[
A_{k-1}(\tilde{Y}_k) =
\argmax_{\tilde{Y}_{k-1} = (k-2)L+1, \ldots, \tilde{Y}_k-L}
m_{k-1}(\tilde{Y}_{k-1}).
\]
and set partial sum maximum $m_k$ as
\[
m_k(\tilde{Y}_k) = 
f(\tilde{Y}_k) + m_{k-1} \Big(A_{k-1}(\tilde{Y}_k)\Big).
\]
\end{enumerate}

\raggedright \emph{Reconstruct centroid $\hat{Y}_c$ using backtrack pointers:}

\begin{enumerate}[Step 1.]
\setcounter{enumi}{2}
\item Set last binding site position:
\[
\hat{Y}_{c,c} = \argmax_{\tilde{Y}_c = (c-1)L+1, \ldots, n-L+1}
m_c(\tilde{Y}_c).
\]
Note that, by construction,
$\max_{\tilde{Y}:c(\tilde{Y})=c} \sum_{k=1}^c f(\tilde{Y}_k) =
m_c(\hat{Y}_{c,c})$.

\item For $k = c, \ldots, 2$ do: recover the remainder of $\hat{Y}_c$ by
setting $\hat{Y}_{c,k-1} = A_{k-1}(\hat{Y}_{c,k})$.
\end{enumerate}
\end{algorithm}

\begin{theorem}
\label{thm:lcent}
Algorithm~\ref{alg:localcentroid} correctly identifies the paired local
centroids
\[
\hat{Y}_c =
\argmin_{\tilde{Y}:c(\tilde{Y}) = c}
\Exp_{Y \given c(Y)=c, R, \Theta} \big[ H_A(\tilde{Y}, Y) \big].
\]
\end{theorem}

\begin{proof}
From Lemma~\ref{lem:lcent} we know that $\hat{Y}_c$ is the argument of
$\max_{\tilde{Y}:c(\tilde{Y})=c} \sum_{k=1}^c f(\tilde{Y}_k)$. The key
device in Algorithm~\ref{alg:localcentroid} is to exploit the recursion in
Equation~\ref{eq:lcent1} to define $m_1(\tilde{Y}_1) = f(\tilde{Y}_1)$ and
\begin{equation}
\label{eq:lcent2}
m_k(\tilde{Y}_k) = 
f(\tilde{Y}_k) + \max_{\tilde{Y}_{k-1} = (k-2)L+1, \ldots, \tilde{Y}_k-L}
m_{k-1}(\tilde{Y}_{k-1}),
\end{equation}
for $k > 1$, to store partial sum maxima. Now it follows that
\[
\max_{\tilde{Y}:c(\tilde{Y})=c} \sum_{k=1}^c f(\tilde{Y}_k)
= \max_{\tilde{Y}_c = (c-1)L+1, \ldots, n-cL+1} m_c(\tilde{Y}_c),
\]
and so Step~3 must be correct. The correctness of Step~4 relies on the right
specification of $m$ in Steps~1 and~2; but these steps are a straightforward
application of Equation~\ref{eq:lcent1} using the definition of $m_1$ and a
formulation of Equation~\ref{eq:lcent2} based on the backtrack pointers $A$,
and so the algorithm is correct.
\end{proof}

We note that the paired local centroids minimize an expected posterior upper
bound $H_A$ on the loss $H$, and so the actual local centroid might not be
attained. We expect, however, that for common cases in which the motif
coverage $c(Y) L$ is much smaller than $n$ that the bound is tight since $H_A$
approximates $H$ well and thus the two local centroids often coincide.

\subsubsection{Global centroid}
While the local centroids already convey information about the distribution of
posterior mass in the space of binding site configurations, the end goal of
the analysis is a point estimate that is, in itself, a good representative of
the space. Following the strategy we outlined in the beginning of this
section, we can further summarize the information in the local centroids by
identifying a configuration $\hat{Y}$ that minimizes the expected conditional
Hamming loss, as in Equation~\ref{eq:gcent}. This approach, however, entails
the same difficulties as defining the centroid based on all points in the
space, and it is thus not treatable by a systematic approach---we are now just
restricting the configurations to the local centroids.

The global centroid can be defined by direct enumeration of all possible
configurations while keeping the minimizer of the expected conditional
posterior loss, but this ``brute-force'' approach considers an exponential
number of solutions. A simple heuristic is to restrict the global centroid to
be one of the local centroids,
\begin{equation}
\label{eq:globalcentroid}
\hat{Y} = \argmin_{\tilde{Y} \in \{\hat{Y}_c\}_{c=0}^C}
\Exp_{c(Y) \given R, \Theta} \big[ H(\hat{Y}_{c(Y)}, \tilde{Y}) \big].
\end{equation}
Another alternative is to just take as global centroid the local centroid of
the modal number of binding sites, $\hat{Y}=\hat{Y}_{c^*}$, where $c^* \doteq
\argmax_{c=0, \ldots, C} \Pr(c(Y)=c\given R, \Theta)$. From now on we adopt
the global centroid in Equation~\ref{eq:globalcentroid} for simplicity and,
again, computational expediency.

Before we continue to our next example, let us remark that a
\emph{constrained}, on the number of binding sites, global centroid might be
more computationally feasible since we are restricting the space of available
configurations. For instance, consider the 1-global centroid,
\[
\hat{Y}_o \doteq
\argmin_{\tilde{Y}:c(\tilde{Y}) = 1}
\Exp_{Y \given R, \Theta} \Big[ H(\tilde{Y}, Y) \Big].
\]
As when defining local centroids, we can approximate $\hat{Y}_o$ using a
paired loss, and since
\[
\begin{split}
\Exp_{Y \given R, \Theta} \Big[ H_A(\tilde{Y}, Y) \Big] & =
\sum_{c=0}^C \sum_{Y:c(Y)=c} \sum_{k=1}^c
H_1(\tilde{Y}, Y_k) \Pr(Y \given R, \Theta) \\
& = 
\sum_{i=1}^n \sum_{c=0}^C \sum_{Y:c(Y)=c} \sum_{k=1}^c
H_1(\tilde{Y}, i) \Pr(Y_k = i \given R, \Theta) \\
& = 
\sum_{i=1}^n H_1(\tilde{Y}, i) 
\sum_{c=0}^C \sum_{Y:c(Y)=c} \sum_{k=1}^c
\Pr(Y_k = i \given R, \Theta) \\
& =
\sum_{i=1}^n H_1(\tilde{Y}, i) P_c(i \given R, \Theta),
\end{split}
\]
where
\begin{equation}
\label{eq:pc}
P_c(i \given R, \Theta) \doteq \sum_{c=1}^C \sum_{Y : c(Y)=c}
\sum_{k=1}^c \Pr(Y_k = i \given R, \Theta),
\end{equation}
we have that
\[
\hat{Y}_o =
\argmin_{\tilde{Y}:c(\tilde{Y})=1}
\Exp_{Y \given R, \Theta} \Big[ H_A(\tilde{Y}, Y) \Big]
= \argmax_{\tilde{Y}:c(\tilde{Y})=1}
G(\tilde{Y}, \cdot) * P_c(\cdot \given R, \Theta).
\]
It is important to note that while the restriction of one binding site might
seem artificial, the derivation of $\hat{Y}_o$ is helpful in recognizing
sequence regions that are likely to host binding sites. In fact, since $P_c$
captures the posterior probability of having a binding site starting at each
position, and considering the overlap gain $G$, the convolution of $G$ and
$P_c$ highlights positions that have higher posterior probability of being
covered by a binding site.


\begin{example}
\label{ex:mbs}
We revisit the same sequence from Example~\ref{ex:sbs}, but now allow for at
most $C = \lfloor n / L \rfloor = 33$ binding sites, and adopt the prior given
in Equation~\ref{eq:priorp} with $b = 3$ and thus $p = 1 - b/n = 0.985$.
Using Algorithm~\ref{alg:marg} we are able to compute the conditional marginal
posteriors $\Pr(c(Y) \given R, \Theta)$ and $\Pr(Y_k \given c(Y), R,
\Theta)$ for $k = 1, \ldots, c(Y)$. These posterior distributions yield the
local centroids---according to Algorithm~\ref{alg:localcentroid}---and the
global centroid from Equation~\ref{eq:globalcentroid}. In Table~\ref{tab:ex2c}
we list the marginal posterior $\Pr(c(Y) = c \given R, \Theta)$ up to the
smallest $c$ such that $\Pr(c(Y) \le c \given R, \Theta) > 0.95$, along with
the local centroids; the global centroid $\hat{Y}_C$ is highlighted.
Interestingly, the global centroid coincides with the local centroid from the
modal number of binding sites.
\begin{table}[bht]
\caption{Centroids and marginal posterior distribution of number of binding
sites. The global centroid and the modal number of binding sites are
highlighted in bold.}
\label{tab:ex2c}
\begin{tabular*}{\textwidth}
{c@{\extracolsep{\fill}}l@{\extracolsep{\fill}}c@{\extracolsep{\fill}}c}
\hline
$c$ & $\hat{Y}_c$ & $\Pr(c(Y)=c \given R, \Theta)$ &
$\Pr(c(Y) \le c \given R, \Theta)$ \\ \hline
$0$ & -- & $0.014$ & $0.014$ \\
$1$ & $36$ & $0.075$ & $0.089$ \\
$2$ & $36,147$ & $0.181$ & $0.270$ \\
$\mathbf{3}$ & $\mathbf{13,36,147}$ & $\mathbf{0.254}$ & $0.524$ \\
$4$ & $13,36,63,147$ & $0.233$ & $0.757$ \\
$5$ & $13,36,63,147,167$ & $0.147$ & $0.904$ \\
$6$ & $3,29,36,63,147,167$ & $0.067$ & $0.971$ \\
\hline
\end{tabular*}
\end{table}

In Figure~\ref{fig:ex2pc} we display the posterior probabilities of binding
site coverage $P_c$ from Equation~\ref{eq:pc}, along with the convolutions
that are needed to define the 1-global centroid $\hat{Y}_o = 36$. As can be
seen, position $36$ has a lot of support, being present in all the local
centroids listed in Table~\ref{tab:ex2c}; in fact, the probability of a
binding site starting at position $36$ is greater than $50\%$.
\begin{figure}[hbtp]
\includegraphics[width=\fracfloat\textwidth]{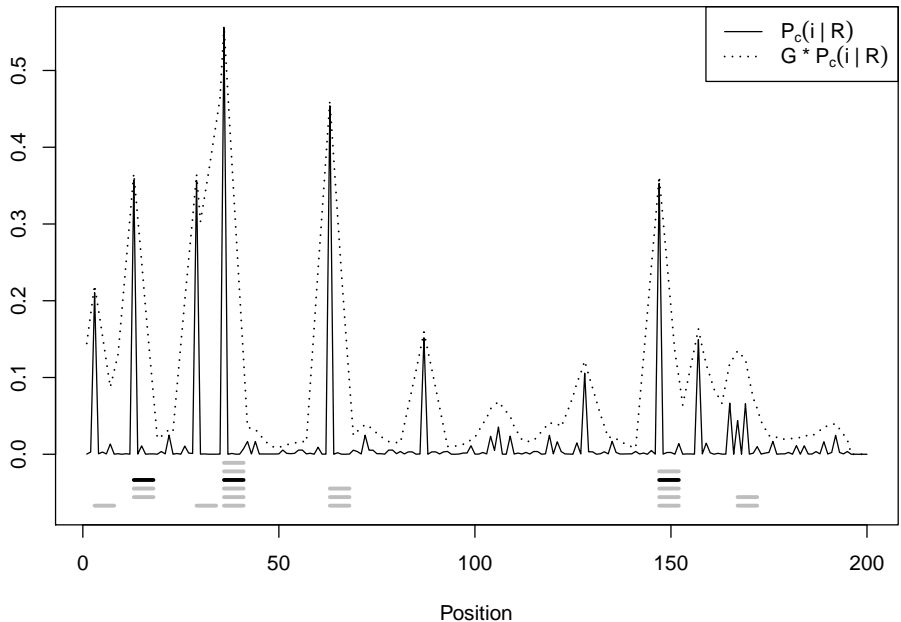}
\caption{Posterior binding site coverage $P_c$ in solid line and convolution
$G * P_c$ in dotted line. Local centroids are listed below in gray; the global
centroid is in black.}
\label{fig:ex2pc}
\end{figure}

While $P_c$ can provide us guidance for which positions are likely to start a
binding site, using $P_c$ to define local centroids can be misleading. For
instance, we could expect that the local centroid with three binding
sites---the modal number of binding sites---would be, following a decreasing
order on $P_c$, $36$, $63$, and $147$. However, if we examine the marginal
posteriors $\Pr(Y_k \given c(Y)=3, R, \Theta)$ in Figure~\ref{fig:ex2pk} we
realize that position $13$ is favored over position $63$ because,
if $F_k \doteq G * \Pr_k(\cdot \given c(Y)=3, R, \Theta)$,
$F_1(13) + F_2(36) > F_1(36) + F_2(63)$.
\begin{figure}[hbtp]
\includegraphics[width=\fracfloat\textwidth]{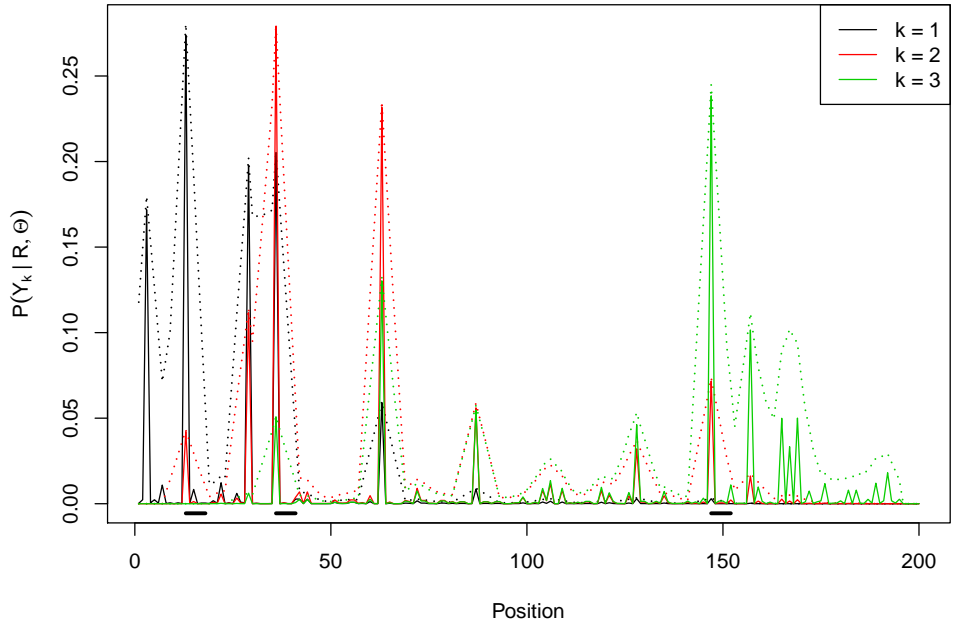}
\caption{Marginal posterior distributions
$\Pr(Y_k \given c(Y) = 3, R, \Theta)$ in solid line and convolutions $G *
\Pr(\cdot \given c(Y), R, \Theta)$ in dotted line. The local centroid is
displayed at the bottom.}
\label{fig:ex2pk}
\end{figure}

\end{example}

\section{Multiple sequences, multiple binding sites per sequence, random
motif}
We are ready to address our model in broader generality: the dataset now
comprises $m$ sequences, $R = \{R_i\}_{i=1}^m$, and thus binding site
configurations are also indexed by sequence, $Y = \{Y_i\}_{i=1}^m$. As before,
we have that $Y$ is independent of motif parameters $\Theta$, but we further
assume that sequences and configurations are conditionally independent given
$\Theta$:
\begin{equation}
\label{eq:rylhood}
\Pr(R, Y \given \Theta) = \prod_{i=1}^m \Pr(R_i, Y_i \given \Theta)
= \prod_{i=1}^m \Pr(R_i \given Y_i,\Theta) \Pr(Y_i).
\end{equation}

Given $\Theta$ we would be able to apply the methods discussed this far to
each sequence separately: compute forward and backward sums to obtain marginal
posterior probabilities for each $Y_i$ and then find local centroids and
the $i$-th global centroid. We will, however, assume that $\Theta$ is random,
\begin{equation}
\label{eq:thetaprior}
\theta_j \sim \text{Dir}(\alpha_j), \quad j = 0, 1, \ldots, L,
\end{equation}
independently, and we thus wish to also conduct inference on the background
and motif compositions. This assumption, albeit more realistic, complicates
matters, since the marginal unconditioned posterior distributions of $Y$ and
$\Theta$ are not readily available; we are now required to estimate them
before obtaining centroid estimates. To this end, we present next a Gibbs
sampler \citep{geman84,liu08} that draws $Y_i$ for each sequence given
$\Theta$ and then samples $\Theta$ conditional on the binding site
configurations $Y$, similar to the approach in~\citep{liu95}.

\subsection{Sampling $\Theta$ given $Y$ and $R$}
Since the prior on $\Theta$ is conjugate, we should be able to sample $\Theta$
exactly from a Dirichlet distribution. From Equations~\ref{eq:rylhood}
and~\ref{eq:thetaprior} we have
\[
\begin{split}
\Pr(\theta_0 \given Y, R) & \propto
\Bigg[ \prod_{i=1}^m \prod_{s \in \mathcal{S}} \prod_{j \in BG_i}
\theta_{0,s}^{I(R_{ij}=s)} \Bigg]
\Bigg[ \prod_{s \in \mathcal{S}} \theta_{0,s}^{\alpha_{0,s}-1} \Bigg] \\
& =
\prod_{s \in \mathcal{S}}
\theta_{0,s}^{\sum_{i=1}^m \sum_{j \in BG_i} I(R_{ij}=s) + \alpha_{0,s}-1},
\end{split}
\]
and so $\theta_0 \given Y, R \sim \text{Dir}(N_0(Y, R) + \alpha_0)$,
where $N_0(Y, R) = \{N_{0,s}\}_{s \in \mathcal{S}}$ and
\[
N_{0,s} = \sum_{i=1}^m \sum_{j \in BG_i} I(R_{ij}=s)
\]
is the number of background positions across all sequences that have symbol
$s$. Similarly, for the $j$-th position in the motif,
\[
\begin{split}
\Pr(\theta_j \given Y, R) & \propto
\Bigg[ \prod_{i=1}^m \prod_{s \in \mathcal{S}} \prod_{k=1}^{|Y_i|}
\theta_{j,s}^{I(R_{i,Y_{ik}+j-1}=s)} \Bigg]
\Bigg[ \prod_{s \in \mathcal{S}} \theta_{j,s}^{\alpha_{j,s}-1} \Bigg] \\
& =
\prod_{s \in \mathcal{S}}
\theta_{0,s}^{\sum_{i=1}^m \sum_{j \in BG_i} I(R_{ij}=s) + \alpha_{0,s}-1},
\end{split}
\]
and thus $\theta_j \given Y, R \sim \text{Dir}(N_j(Y, R) + \alpha_j)$, with
$N_j(Y, R) = \{N_{j,s}\}_{s \in \mathcal{S}}$ and
\[
N_{j,s} = \sum_{i=1}^m \sum_{k=1}^{|Y_i|} I(R_{i,Y_{ik} + j - 1} = s)
\]
is the number of motif $j$-th positions across all sequences and binding sites
that have symbol $s$.

\subsection{Sampling $Y_i$ given $\Theta$ and $R$} 
Each configuration $Y_i$ for the $i$-th sequence is conditionally independent
given $\Theta$, so we can devise a sampling procedure that can be applied to
each sequence in turn. To simplify the notation, let us drop the sequence
index in what follows, that is, $Y_i$ is $Y$, $R_i$ is $R$, and so on. We will
be following a similar approach to Sections~\ref{ssec:postcy}
and~\ref{ssec:postykcy}, but instead of summing to obtain marginal
distributions we will be sampling \emph{exactly}.

To sample from the conditional posterior on $Y$, we first sample $c(Y)=c$
according to Equation~\ref{eq:cY} and then proceed to sample $Y$ from its last,
$c$-th binding site up to its first binding site. For this reason, this
strategy is commonly referred to as ``stochastic backtracking'', since it can
be regarded as a stochastic version of Step~4 in Algorithms~\ref{alg:marg}
and~\ref{alg:localcentroid}. Sampling $Y$ is similar to the predictive update
step in \citep{liu95}, which, on its turn, is based on a stochastic variation
of expectation-maximization where missing data is imputed \citep{tanner87};
however, here we exploit a hierarchical structure on $c(Y)$ and do not use the
collapsing technique of \citet{liu94}.

Exploiting the conditional independence of the sequence configurations and
Equation~\ref{eq:Yprior1} the last binding site can be sampled using
\begin{equation}
\label{eq:sampleYlast}
\begin{split}
\Pr(Y_c \given c(Y), R, \Theta) & =
\frac{\sum_{Y_1, \ldots, Y_{c-1}} \Pr(R \given Y, \Theta)}
{\sum_{\tilde{Y}_c} \sum_{\tilde{Y}_1, \ldots, \tilde{Y}_{c-1}}
\Pr(R \given \tilde{Y}, \Theta)} \\
& =
\frac{F_{c-1,Y_c-1} \lambda(Y_c; \Theta)}
{\sum_{\tilde{Y}_c = (c-1)L + 1}^{n-L+1}
F_{c-1,\tilde{Y}_c-1} \lambda(\tilde{Y}_c; \Theta)}.
\end{split}
\end{equation}

To sample the (intermediate) $j$-th binding site we use a similar expression:
\begin{equation}
\label{eq:sampleY}
\begin{split}
\Pr(Y_j \given Y_{j+1}, \ldots, Y_c, c(Y), R, \Theta) & =
\frac{\Pr(Y_j, \ldots, Y_c, c(Y), R, \Theta)}
{\sum_{\tilde{Y}_j} \Pr(\tilde{Y}_j, \ldots, \tilde{Y}_c, c(Y), R, \Theta)} \\
& =
\frac{\sum_{Y_1, \ldots, Y_{j-1}} \Pr(R \given Y, \Theta)}
{\sum_{\tilde{Y}_j} \sum_{\tilde{Y}_1, \ldots, \tilde{Y}_{j-1}}
\Pr(R \given \tilde{Y}, \Theta)} \\
& =
\frac{F_{j-1, Y_j-1} \lambda(Y_j; \Theta)}
{\sum_{\tilde{Y}_j=(j-1)L+1}^{Y_{j+1}-L}
F_{j-1, \tilde{Y}_j-1} \lambda(\tilde{Y}_j; \Theta)}.
\end{split}
\end{equation}
By making the convention that $Y_{c+1}=|R|+1$ we can reduce
Equation~\ref{eq:sampleYlast} to Equation~\ref{eq:sampleY}. Moreover, note
that Equation~\ref{eq:sampleY} implies that
\[
\Pr(Y_j \given Y_{j+1}, \ldots, Y_c, c(Y), R, \Theta) =
\Pr(Y_j \given Y_{j+1}, c(Y), R, \Theta),
\]
as expected.

We summarize the whole procedure in Algorithm~\ref{alg:gibbs}. Note how
Steps~1.1 to~1.3 are analogous to Steps~1 to~3 in Algorithm~\ref{alg:marg},
and how Step~1.4 is an stochastic version of Step~4 in
Algorithm~\ref{alg:marg}: as previously stated, we are now sampling backwards
instead of summing backwards. To obtain the centroids we follow the procedure
described in Section~\ref{ssec:postinf}, but adopting Monte Carlo
estimates of the marginal posterior distributions, for $i = 1, \ldots, m$,
\[
\begin{split}
\widehat{\Pr}(c(Y_i) = c \given R) & \approx
\frac{1}{T} \sum_{t=1}^T I \big( c(Y_i^{(t)}) = c \big), \\
\widehat{\Pr}(Y_{ik} = j \given c(Y_i) = c, R) & \approx
\frac{\sum_{t=1}^T I \big( Y_{ik}^{(t)} = j \big)
I\big( c(Y_i^{(t)}) = c \big)}
{\sum_{t=1}^T I\big( c(Y_i^{(t)}) = c \big)},
\quad k = 1, \ldots, c, \\
\end{split}
\]
where $T$ is the number of samples.

\begin{algorithm}[htpb]
\caption{Gibbs sampler for $\Pr(Y, \Theta \given R)$.}
\label{alg:gibbs}
\raggedright \emph{Set $\Theta^{(0)}$ arbitrarily. For $t = 1, \ldots$ (until
convergence) do:}
\begin{enumerate}[Step 1.]
\item \emph{(Sample $Y \given \Theta, R$)} For each sequence $i = 1, \ldots,
m$, do: let $n = |R_i|$, $C = \lfloor n/L \rfloor$ and sample $Y_i \given
R_i, \Theta$.

\begin{enumerate}[Step {1}.1.]
  \item \emph{(Initialize)} Set $F_{0,j} = 1$ for $j = 0, 1, \ldots, n$ and
  for $c = 1, \ldots, C$ set $F_{c,j} = 0$ when $j < cL$.

  \item \emph{(Compute forward sums)} For $c = 1, \ldots, C$ and $j = cL+1,
  \ldots, n$ do: set $F_{c,j}$ as in Equation~\ref{eq:fcjrec},
  \[
  F_{c,j} = F_{c,j-1} + F_{c-1,j-L} \lambda_i(j-L+1; \Theta^{(t-1)}),
  \]
  where $\lambda_i$ uses $R_i$.

  \item \emph{(Sample $c(Y_i^{(t)}) \given R_i, \Theta^{(t-1)}$)} For $c = 0,
  \ldots, C$ do: compute marginal posterior $c(Y_i)$ as in
  Equation~\ref{eq:cY} when applied to the $i$-th sequence,
  \[
  \Pr(c(Y_i) = c \given R_i, \Theta^{(t-1)}) =
  \frac{F_{c,n} \binom{n-c(L-1)}{c}^{-1} \Pr(c(Y_i) = c)}
  {\sum_{\tilde{c}=0}^C F_{\tilde{c},n}
  \binom{n-\tilde{c}(L-1)}{\tilde{c}}^{-1} \Pr(c(Y_i) = \tilde{c})}
  \]
  and sample $c^{(t)} \doteq c(Y_i^{(t)})$ according to
  $\Pr(c(Y_i) = c \given R_i, \Theta^{(t-1)})$.

  \item \emph{(Sample $Y_i^{(t)} \given c(Y_i^{(t)})=c^{(t)}, R_i,
  \Theta^{(t-1)}$)} For $k = c^{(t)}, \ldots, 1$ do: sample $Y_{ik}^{(t)}$
  proportional to
  $F_{k-1, Y_{ik}^{(t)} - 1} \lambda_i(Y_{ik}^{(t)}; \Theta^{(t-1)})$ as in
  Equation~\ref{eq:sampleY},
  \begin{multline*}
  \Pr(Y_{ik}^{(t)} \given Y_{i,k+1}^{(t)},
  c(Y_i^{(t)}) = c^{(t)}, R_i, \Theta^{(t-1)}) = \\
  = \frac{F_{k-1, Y_{ik}^{(t)}-1} \lambda_i(Y_{ik}^{(t)}; \Theta^{(t-1)})}
  {\sum_{\tilde{Y}_k=(k-1)L+1}^{Y_{i,k+1}^{(t)}-L}
  F_{k-1, \tilde{Y}_k-1} \lambda_i(\tilde{Y}_k; \Theta^{(t-1)})}
  \end{multline*}
  \end{enumerate}

\item \emph{(Sample $\Theta \given Y, R$)} For $j = 0, \ldots, L$ compute
$N_j(Y^{(t)}, R)$ and then sample
$\theta_j^{(t)} \given Y^{(t)}, R \sim \text{Dir}(N_j(Y^{(t)}, R) + \alpha_j)$.
\end{enumerate}
\end{algorithm}

\begin{example}
For the random motif version of Example~\ref{ex:mbs} we simulate $m = 20$
sequences of same length $n = 200$ using $\Theta$ from Table~\ref{tab:theta}
and the prior for $Y_i$, $i = 1, \ldots, m$, from Equation~\ref{eq:priorp}
with $p = 1 - 1/n = 0.995$.

We continue focusing on the inference of binding site configurations in the
same sequence from previous examples, which is the first sequence in the
simulated dataset. We assume a non-informative prior on $\Theta$ by setting
$\alpha_{j,s} = 1$ for $s \in \mathcal{S}$ and $j = 0, \ldots, L$; the prior
on each sequence $Y_i$ is the same prior from Example~\ref{ex:mbs} with $p =
0.985$. Algorithm~\ref{alg:gibbs} is run for $10,\!000$ iterations to
guarantee convergence (diagnostics not shown.)

The marginal posterior distribution of $\Theta$ can be assessed in
Figure~\ref{fig:ex3theta}. Since most positions in the sequences are
background sequences $\theta_0$ has very small posterior variances. Also note
that the canonical palindromic E-box motif, with consensus \texttt{CACGTG}, is
recovered.

\begin{figure}[hbt]
\includegraphics[width=\fracfloat\textwidth]{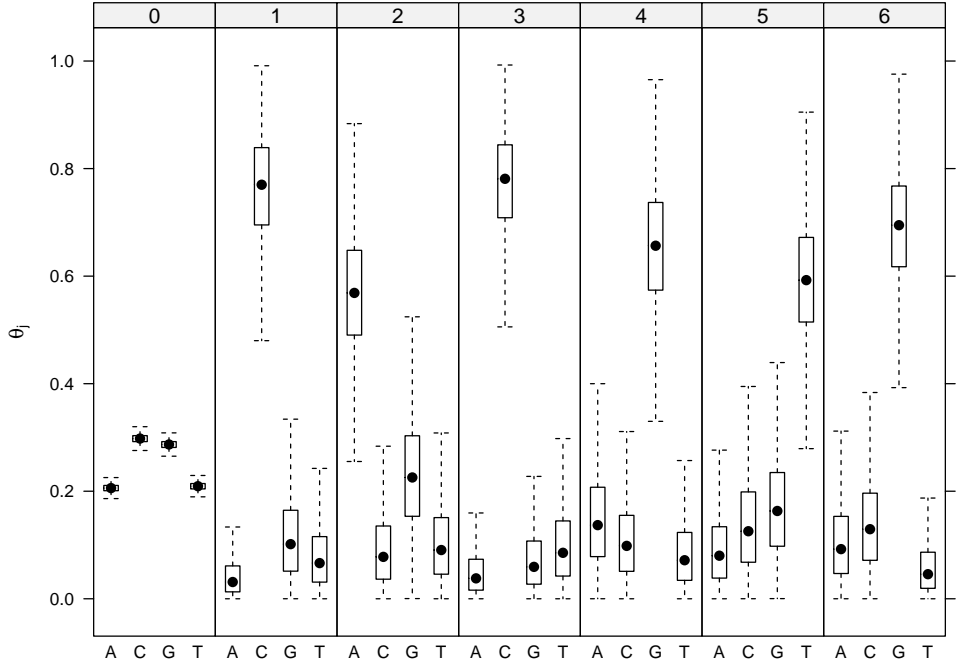}
\caption{Boxplots of MCMC samples for $\Theta$ (outliers are not shown.)}
\label{fig:ex3theta}
\end{figure}

The procedure is now similar to what we presented in Example~\ref{ex:mbs}; the
main difference is that the marginal posterior distributions are estimated
from the MCMC samples. Table~\ref{tab:ex3c} lists the estimated marginal
posterior distribution of the number of binding sites, the local and global
centroids. The global centroid does not coincide with the local centroid for
the modal number of binding sites. Moreover, the local centroids here are
different from the (conditional) local centroids in Example~\ref{ex:mbs}, most
likely due to the randomness of $\Theta$ being taken into account.

\begin{table}[bht]
\caption{Centroids and estimated marginal posterior distribution of number of
binding sites. The global centroid and the modal number of binding sites are
highlighted in bold.}
\label{tab:ex3c}
\begin{tabular*}{\textwidth}
{c@{\extracolsep{\fill}}l@{\extracolsep{\fill}}c@{\extracolsep{\fill}}c}
\hline
$c$ & $\hat{Y}_c$ & $\widehat{\Pr}(c(Y)=c \given R, \Theta)$ &
$\widehat{\Pr}(c(Y) \le c \given R, \Theta)$ \\ \hline
$0$ & -- & $0.026$ & $0.026$ \\
$1$ & $29$ & $0.107$ & $0.133$ \\
$\mathbf{2}$ & $\mathbf{29,167}$ & $0.210$ & $0.343$ \\
$3$ & $29,63,167$ & $\mathbf{0.274}$ & $0.617$ \\
$4$ & $13,36,147,167$ & $0.201$ & $0.818$ \\
$5$ & $13,29,63,147,167$ & $0.120$ & $0.938$ \\
$6$ & $13,29,36,63,147,167$ & $0.046$ & $0.984$ \\
\hline
\end{tabular*}
\end{table}

Figure~\ref{fig:ex3pc} displays the estimated $P_c$, $G * P_c$, and the
centroids. We see that compared to Example~\ref{ex:mbs} some posterior mass
has shifted to positions $29$ and to the group of positions $166$, $167$, and
$168$. Here we clearly see the advantage of a centroid estimator: $G * P_c$,
and later $G * \Pr_k(\cdot \given R)$, gathers evidence of motif binding from
nearby positions, yielding a better summary---according to our choice of loss
function---of the distribution of posterior mass.

\begin{figure}[hbt]
\includegraphics[width=\fracfloat\textwidth]{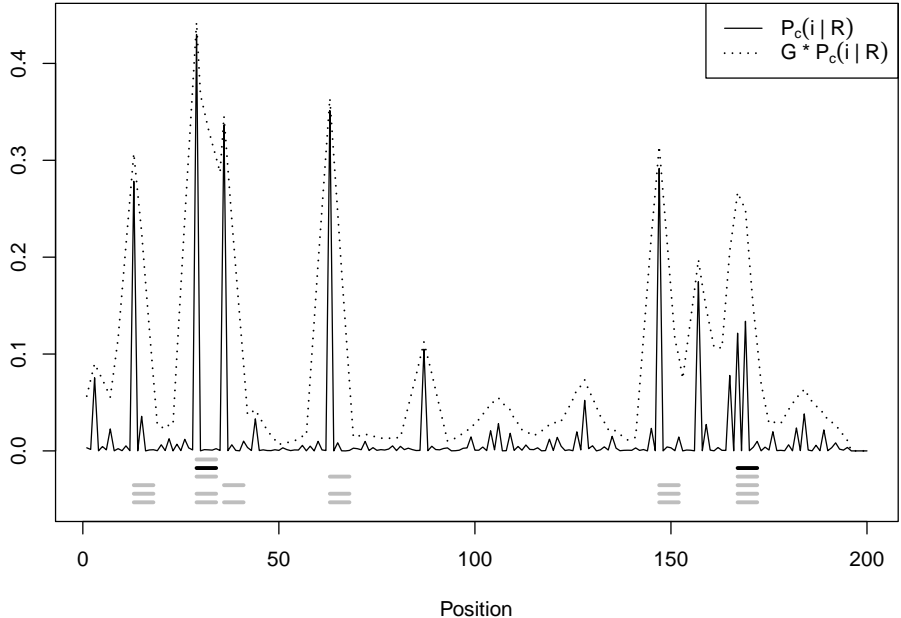}
\caption{Estimated posterior binding site coverage $P_c$ in solid line and
convolution $G * P_c$ in dotted line. Local centroids are listed below in
gray; the global centroid is in black.}
\label{fig:ex3pc}
\end{figure}

The selection of position $167$ in the second local centroid $\hat{Y}_2$ might
seem puzzling since the peaks at positions $36$, $63$, and $147$ hold higher
coverage probabilities. Checking $\widehat{\Pr}(Y_k \given R)$ in
Figure~\ref{fig:ex3pk} helps dismiss any doubts: most of the support for these
positions come from configurations with higher number of binding sites, as
evidenced by the respective local centroids, but these configurations hold
relatively low posterior mass. When $c(Y)=2$, the prior on $Y_{2,2}$ assigns
more posterior probability to higher positions, close to the end of the
sequence, simply because there are more configurations for $Y_{2,2}$ on these
positions. It is also important to notice that while none of the positions in
the cluster $166$--$168$ has higher marginal posterior mass than positions
$63$ and $147$, the convolution $G * \widehat{\Pr}_2(\cdot \given R)$ is
maximized at position $167$, that is, the cluster when taken together has more
support from the data, as weighted by $G$.

\begin{figure}[hbt]
\includegraphics[width=\fracfloat\textwidth]{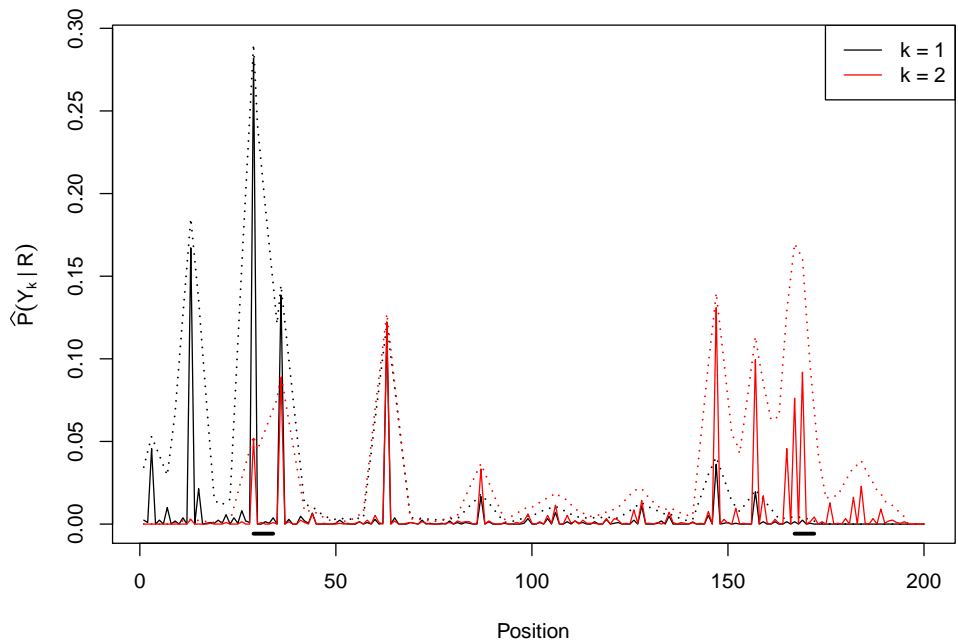}
\caption{Estimated marginal posterior distributions
$\Pr(Y_k \given c(Y) = 2, R, \Theta)$ in solid line and convolutions
$G * \Pr(\cdot \given c(Y), R, \Theta)$ in dotted line. The local centroid is
displayed at the bottom.}
\label{fig:ex3pk}
\end{figure}

\end{example}

\begin{example}
We end this section with an example from the real-world dataset in
\citep{tompa05}, sequence set \texttt{yst02r}. The dataset contains $m = 4$
sequences each with $n = 500$ letters. We set $L=16$ and adopt a
non-informative prior on $\Theta$, as in the previous example, and the prior
on each $Y_i$, for the $i$-th sequence, from Equation~\ref{eq:priorp} with $b
= 3$ per thousand positions, so $p = 1 - 3 / 1000 = 0.997$. As in the previous
example, $10,\!000$ iterations suffice to reach convergence.

Let us focus on the second sequence. Figure~\ref{fig:ex4pc} pictures the
binding site coverage probabilities, along with the local centroids. The
global centroid $\hat{Y}_C = \{85,105,169\}$ contains three binding sites, and
it is also the local centroid for the modal number of binding sites, with
$\widehat{\Pr}(c(Y)=3 \given R) = 0.32$. Since most of the posterior mass in
concentrated in configurations with $c(Y) = 3$, the posterior profiles
$\widehat{\Pr}(Y_k \given c(Y)=3, R)$ are similar to $P_c$ and are thus
omitted.

\begin{figure}[hbt]
\includegraphics[width=\fracfloat\textwidth]{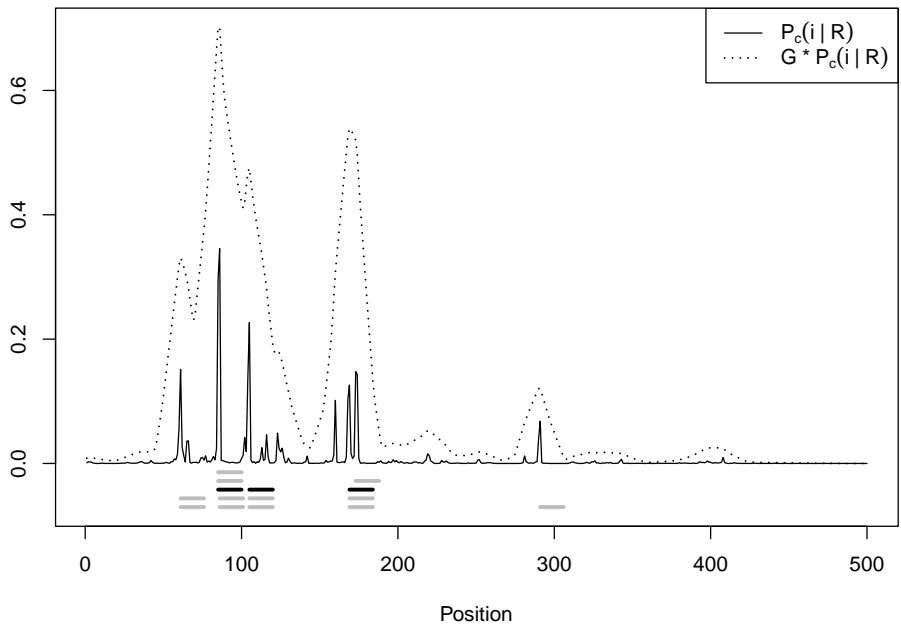}
\caption{Estimated posterior binding site coverage $P_c$ in solid line and
convolution $G * P_c$ in dotted line for real-world dataset, second sequence.
Local centroids are listed below in gray; the global centroid is in black.}
\label{fig:ex4pc}
\end{figure}


From the MCMC samples we can produce the MAP estimate
$\hat{Y}_M = \{86, 105, 174\}$ as the configuration with highest frequency
among the samples: $\widehat{\Pr}(\hat{Y}_M \given R) = 0.032$. In fact, we
can estimate the posterior probability of each sampled binding site
configuration and then, using classic multidimensional scaling
\citep{gower66}, visualize the estimated posterior distribution in
Figure~\ref{fig:ex4dist}. It is interesting to note that the null
configuration---that is, without binding sites---is also very likely with
posterior probability $0.024$. In contrast, the global centroid has very small
posterior probability, close to $0.001$; it sits, however, closer to
configurations with high posterior mass, including the local centroids with
one, two, and four binding sites.

\begin{figure}[hbt]
\includegraphics[width=\fracfloat\textwidth]{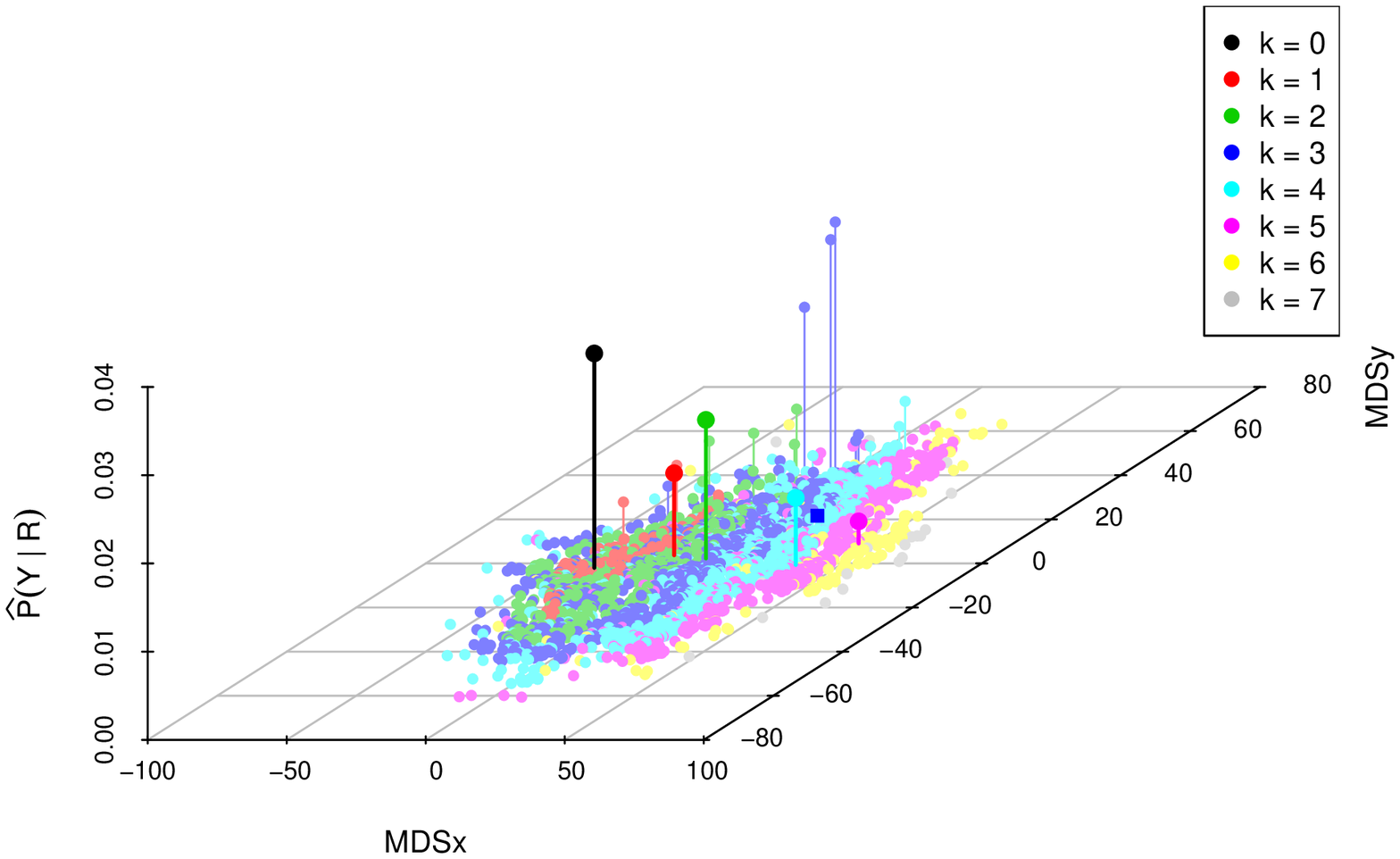}
\caption{Estimated posterior distribution of configurations $Y$ based on MCMC
samples and projected using multidimensional scaling. The colors code
configurations with different number of binding sites. Bold points mark local
centroids, while a square (bold) point highlights the global centroid.}
\label{fig:ex4dist}
\end{figure}

To better assess how the centroid estimator is closer to a mean than a mode
estimator, we plot the estimated posterior distribution of the generalized
loss function $H$ centered at both $\hat{Y}_C$ and $\hat{Y}_M$ in
Figure~\ref{fig:ex4hist}. Since $\Exp_{Y \given R}[H(\hat{Y}_M,Y)] = 42.40$
and $\Exp_{Y \given R}[H(\hat{Y}_C,Y)] = 40.22$, we see that the binding sites
in the centroid configuration are, on average, overlapping two extra positions
with the binding sites in all the configurations when compared to the MAP
estimate's binding sites. Both estimates are fairly similar, but the centroid
reminds us that placing the third binding site at position $169$, instead of
$174$, yields an unlikely configuration, but with a higher chance of
overlapping with binding sites in positions $160$--$175$ that have high
posterior probability. In the context of Figures~\ref{fig:ex4dist}
and~\ref{fig:ex4hist}, the centroid places itself between two clusters that
concentrate posterior mass: one with configurations $Y$ such that
$25 \le H(\hat{Y}_C, Y) \le 40$ and another with configurations further away,
satisfying $40 \le H(\hat{Y}_C, Y) \le 50$.

\begin{figure}[hbt]
\includegraphics[width=\fracfloat\textwidth]{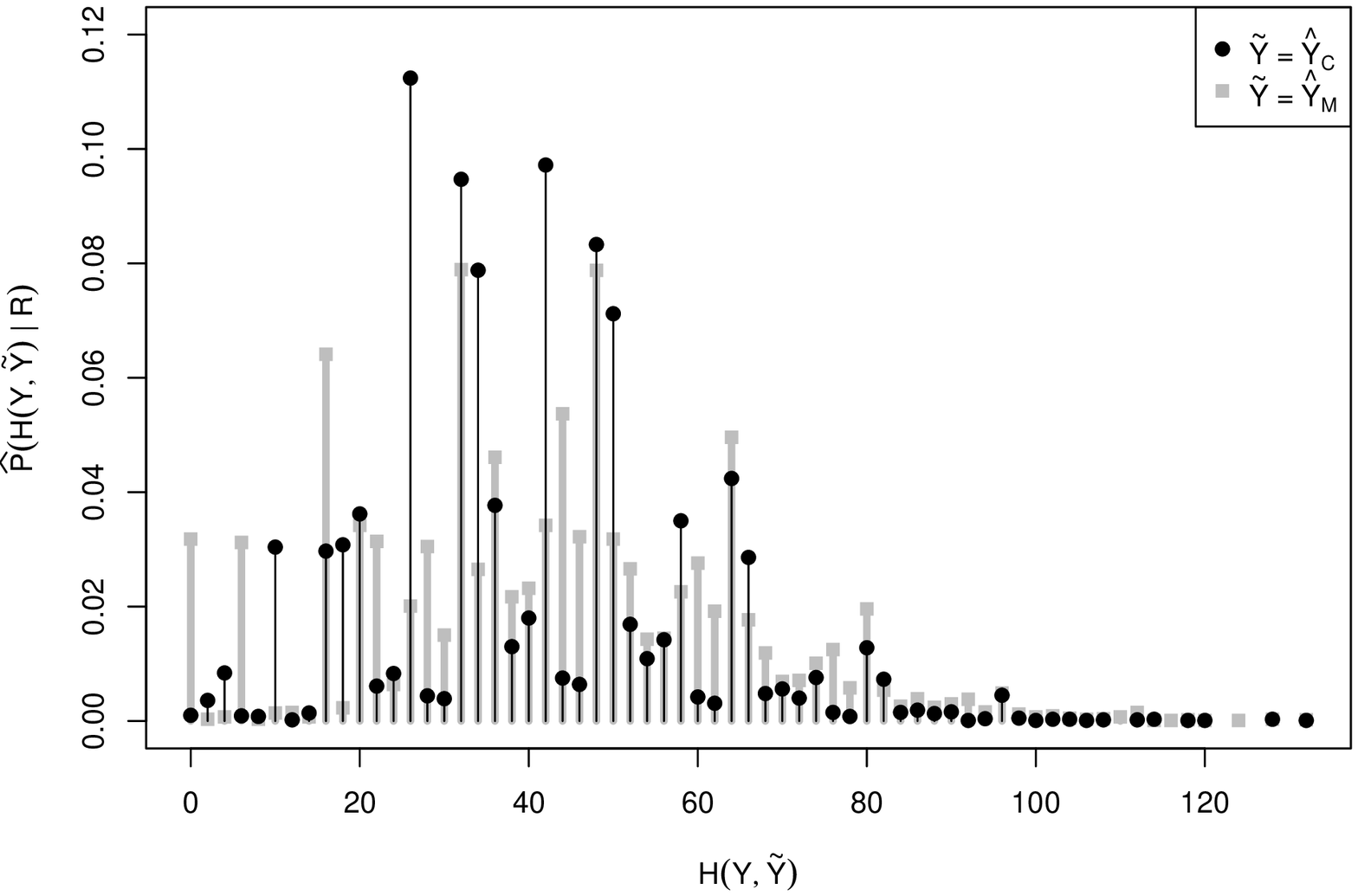}
\caption{Estimated posterior distribution of loss function centered at
$\tilde{Y}$ for the MAP ($\tilde{Y} = \hat{Y}_M$) and centroid ($\tilde{Y} =
\hat{Y}_C$) estimates.}
\label{fig:ex4hist}
\end{figure}

\end{example}

\section{Discussion}
In this paper we have presented a Bayesian approach, similar to the Gibbs
motif sampler in \citep{lawrence93,liu95}, that jointly models motif and
background compositions and binding site locations in a set of sequences. More
importantly, we discuss and formalize an inferential procedure based on the
centroid estimator proposed by \citet{carvalho08}. As in any Bayesian
analysis, we wish to evaluate features of interest in a model based on their
posterior distribution; however, if we are required to pick a representative
configuration, a point in the parameter space, then a principled approach is
to elect a loss function and conduct formal statistical decision analysis. In
this sense, by exploring a more refined loss function that depends on
position-wise comparisons between sequence states---background or motif
positions---we are able to identify a better representative of the posterior
space of binding site configurations. As pointed out in \citep{carvalho08},
the centroid estimator better accounts for the distribution of posterior mass;
it is more similar to a median than to a mode, and can thus offer better
predictive resolution than the MAP estimator \citep{barbieri04}. When applied
to motif discovery, the centroid estimator captures information in the
vicinity of binding site positions through a convolution in marginal posterior
distributions of binding sites.

Given the combinatorial number of possible configurations in the parameter
space it is not feasible to identify the centroid estimate through enumeration
or even a systematic approach. Yet, we devise an approximative scheme that
efficiently optimizes an upper bound on the posterior expected loss and thus
provides a related centroid. Despite its heuristic nature, the proposed method
has another advantage besides computational convenience: it allows for an
informative depiction of the posterior distribution on binding site
configurations. First, when defining the local centroids, we are able to
assess the contributions from each binding site through their marginal
posterior distributions conditional on the number of binding sites, and, in
particular, through the convolution of these marginal profiles with the gain
filter; secondly, when finding the global centroid we explore the marginal
posterior distribution on the number of binding sites. Moreover, other
representations might be helpful in understanding the distribution of
posterior mass, as in the use of $P_c$ (in Equation~\ref{eq:pc}) to pinpoint
the 1-global centroid and measure the overall support of the configurations to
a binding site at some specific position in the sequence. These comments are
in the spirit of an estimator being also a communicator of the posterior space
and the particular choice of prior distribution
\citep[see][Section 4.10]{berger85}.

It is important to note that even when the model is accurate, a poor inference
might fail in recovering relevant features of the space. In Example~2, the MAP
estimate is the null configuration, while the centroid indicates three binding
sites that represent a group of configurations that jointly pool significant
posterior mass. It is also common that the posterior distribution is too
complex to be reasonably captured by a single representative; in this case the
expected posterior loss could also be used to partition the space and further
define additional representatives as conditional estimates on each subspace.
This is a direction of work that warrants interest and that we intend to
follow next.

Further improvements can be obtained by specifying a more complex model that
accounts, for example, for higher order Markov chains with more states for the
background, as in \citep{roth98,liu01}, phylogenetic profiles
\citep{newberg07}, structural information \citep{xing04}, a variable motif
length, or dependency among motif positions. As pointed out by \citet{hu05},
motif discovery using sequence only is well known for low signal-to-noise
ratio; future extensions would also incorporate other data sources, such as
gene expression or ChIP-Seq data, to increase the signal-to-noise ratio.

\section*{Acknowledgements}
The author would like to thank Antonio Gomes for the helpful discussions and
comments in the text.

\bibliographystyle{chicago}
\bibliography{motif}

\begin{thebibliography}{}

\bibitem[\protect\citeauthoryear{Bailey and Elkan}{Bailey and
  Elkan}{1995}]{bailey95}
Bailey, T. and C.~Elkan (1995).
\newblock Unsupervised learning of multiple motifs in biopolymers using
  expectation maximization.
\newblock {\em Machine learning\/}~{\em 21\/}(1), 51--80.

\bibitem[\protect\citeauthoryear{Barbieri and Berger}{Barbieri and
  Berger}{2004}]{barbieri04}
Barbieri, M. and J.~Berger (2004).
\newblock Optimal predictive model selection.
\newblock {\em The Annals of Statistics\/}~{\em 32\/}(3), 870--897.

\bibitem[\protect\citeauthoryear{Berger}{Berger}{1985}]{berger85}
Berger, J. (1985).
\newblock {\em Statistical decision theory and Bayesian analysis}.
\newblock Springer.

\bibitem[\protect\citeauthoryear{Besag}{Besag}{1986}]{besag86}
Besag, J. (1986).
\newblock On the statistical analysis of dirty pictures.
\newblock {\em Journal of the Royal Statistical Society. Series B
  (Methodological)\/}~{\em 48\/}(3), 259--302.

\bibitem[\protect\citeauthoryear{Carvalho and Lawrence}{Carvalho and
  Lawrence}{2008}]{carvalho08}
Carvalho, L. and C.~Lawrence (2008).
\newblock Centroid estimation in discrete high-dimensional spaces with
  applications in biology.
\newblock {\em Proceedings of the National Academy of Sciences of the United
  States of America\/}~{\em 105\/}(9), 3209.

\bibitem[\protect\citeauthoryear{Dempster, Laird, and Rubin}{Dempster
  et~al.}{1977}]{dempster77}
Dempster, A., N.~Laird, and D.~Rubin (1977).
\newblock Maximum likelihood from incomplete data via the {EM} algorithm.
\newblock {\em Journal of the Royal Statistical Society. Series B
  (Methodological)\/}~{\em 39\/}(1), 1--38.

\bibitem[\protect\citeauthoryear{Ding, Chan, and Lawrence}{Ding
  et~al.}{2005}]{ding05}
Ding, Y., C.~Chan, and C.~Lawrence (2005).
\newblock {RNA} secondary structure prediction by centroids in a {Boltzmann}
  weighted ensemble.
\newblock {\em {RNA}\/}~{\em 11\/}(8), 1157--1166.

\bibitem[\protect\citeauthoryear{Geman and Geman}{Geman and
  Geman}{1984}]{geman84}
Geman, S. and D.~Geman (1984).
\newblock Stochastic relaxation, {Gibbs} distributions, and the {Bayesian}
  restoration of images.
\newblock {\em Pattern Analysis and Machine Intelligence, IEEE Transactions
  on\/}~{\em 6\/}(6), 721--741.

\bibitem[\protect\citeauthoryear{Gower}{Gower}{1966}]{gower66}
Gower, J. (1966).
\newblock Some distance properties of latent root and vector methods used in
  multivariate analysis.
\newblock {\em Biometrika\/}~{\em 53\/}(3-4), 325--338.

\bibitem[\protect\citeauthoryear{GuhaThakurta}{GuhaThakurta}{2006}]{guhathakurta06}
GuhaThakurta, D. (2006).
\newblock Computational identification of transcriptional regulatory elements
  in {DNA} sequence.
\newblock {\em Nucleic acids research\/}~{\em 34\/}(12), 3585--3598.

\bibitem[\protect\citeauthoryear{Hu, Li, and Kihara}{Hu et~al.}{2005}]{hu05}
Hu, J., B.~Li, and D.~Kihara (2005).
\newblock Limitations and potentials of current motif discovery algorithms.
\newblock {\em Nucleic acids research\/}~{\em 33\/}(15), 4899--4913.

\bibitem[\protect\citeauthoryear{Lawrence, Altschul, Boguski, Liu, Neuwald, and
  Wootton}{Lawrence et~al.}{1993}]{lawrence93}
Lawrence, C., S.~Altschul, M.~Boguski, J.~Liu, A.~Neuwald, and J.~Wootton
  (1993).
\newblock Detecting subtle sequence signals: a {Gibbs} sampling strategy for
  multiple alignment.
\newblock {\em Science\/}~{\em 262\/}(5131), 208--214.

\bibitem[\protect\citeauthoryear{Liu}{Liu}{1994}]{liu94}
Liu, J. (1994).
\newblock The collapsed {Gibbs} sampler in {Bayesian} computations with
  applications to a gene regulation problem.
\newblock {\em Journal of the American Statistical Association\/}~{\em
  89\/}(427), 958--966.

\bibitem[\protect\citeauthoryear{Liu}{Liu}{2008}]{liu08}
Liu, J. (2008).
\newblock {\em Monte Carlo strategies in scientific computing}.
\newblock Springer Verlag.

\bibitem[\protect\citeauthoryear{Liu, Neuwald, and Lawrence}{Liu
  et~al.}{1995}]{liu95}
Liu, J., A.~Neuwald, and C.~Lawrence (1995).
\newblock {Bayesian} models for multiple local sequence alignment and {Gibbs}
  sampling strategies.
\newblock {\em Journal of the American Statistical Association\/}~{\em
  90\/}(432), 1156--1170.

\bibitem[\protect\citeauthoryear{Liu, Brutlag, and Liu}{Liu
  et~al.}{2001}]{liu01}
Liu, X., D.~Brutlag, and J.~Liu (2001).
\newblock {BioProspector}: discovering conserved {DNA} motifs in upstream
  regulatory regions of co-expressed genes.
\newblock In {\em Pac Symp Biocomput}, Volume~6, pp.\  127--138.

\bibitem[\protect\citeauthoryear{MacIsaac and Fraenkel}{MacIsaac and
  Fraenkel}{2006}]{macisaac06}
MacIsaac, K. and E.~Fraenkel (2006).
\newblock Practical strategies for discovering regulatory {DNA} sequence
  motifs.
\newblock {\em PLoS computational biology\/}~{\em 2\/}(4), e36.

\bibitem[\protect\citeauthoryear{Murrea, McCawa, and Baltimorea}{Murrea
  et~al.}{1989}]{murrea89}
Murrea, C., P.~S. McCawa, and D.~Baltimorea (1989).
\newblock A new {DNA} binding and dimerization motif in immunoglobulin enhancer
  binding, daughterless, {MyoD}, and {Myc} proteins.
\newblock {\em Cell\/}~{\em 56\/}(5), 777--783.

\bibitem[\protect\citeauthoryear{Neuwald, Liu, and Lawrence}{Neuwald
  et~al.}{1995}]{neuwald95}
Neuwald, A., J.~Liu, and C.~Lawrence (1995).
\newblock {Gibbs} motif sampling: detection of bacterial outer membrane protein
  repeats.
\newblock {\em Protein science\/}~{\em 4\/}(8), 1618--1632.

\bibitem[\protect\citeauthoryear{Newberg, Thompson, Conlan, Smith, McCue, and
  Lawrence}{Newberg et~al.}{2007}]{newberg07}
Newberg, L., W.~Thompson, S.~Conlan, T.~Smith, L.~McCue, and C.~Lawrence
  (2007).
\newblock A phylogenetic {Gibbs} sampler that yields centroid solutions for
  cis-regulatory site prediction.
\newblock {\em Bioinformatics\/}~{\em 23\/}(14), 1718--1727.

\bibitem[\protect\citeauthoryear{Pavesi, Mereghetti, Mauri, and Pesole}{Pavesi
  et~al.}{2004}]{pavesi04}
Pavesi, G., P.~Mereghetti, G.~Mauri, and G.~Pesole (2004).
\newblock {Weeder} {Web}: discovery of transcription factor binding sites in a
  set of sequences from co-regulated genes.
\newblock {\em Nucleic acids research\/}~{\em 32\/}(suppl 2), W199--W203.

\bibitem[\protect\citeauthoryear{Pevzner, Sze, et~al.}{Pevzner
  et~al.}{2000}]{pevzner00}
Pevzner, P., S.~Sze, et~al. (2000).
\newblock Combinatorial approaches to finding subtle signals in {DNA}
  sequences.
\newblock In {\em Proceedings of the Eighth International Conference on
  Intelligent Systems for Molecular Biology}, Volume~8, pp.\  269--278.

\bibitem[\protect\citeauthoryear{R{\'e}gnier and Denise}{R{\'e}gnier and
  Denise}{2004}]{regnier04}
R{\'e}gnier, M. and A.~Denise (2004).
\newblock Rare events and conditional events on random strings.
\newblock {\em Discrete Mathematics and Theoretical Computer Science\/}~{\em
  6\/}(2), 191--214.

\bibitem[\protect\citeauthoryear{Roth, Hughes, Estep, and Church}{Roth
  et~al.}{1998}]{roth98}
Roth, F., J.~Hughes, P.~Estep, and G.~Church (1998).
\newblock Finding {DNA} regulatory motifs within unaligned noncoding sequences
  clustered by whole-genome {mRNA} quantitation.
\newblock {\em Nature biotechnology\/}~{\em 16\/}(10), 939--945.

\bibitem[\protect\citeauthoryear{Sandve and Drablos}{Sandve and
  Drablos}{2006}]{sandve06}
Sandve, G. and F.~Drablos (2006).
\newblock A survey of motif discovery methods in an integrated framework.
\newblock {\em Biol Direct\/}~{\em 1\/}(11).

\bibitem[\protect\citeauthoryear{Stormo}{Stormo}{2000}]{stormo00}
Stormo, G. (2000).
\newblock {DNA} binding sites: representation and discovery.
\newblock {\em Bioinformatics\/}~{\em 16\/}(1), 16--23.

\bibitem[\protect\citeauthoryear{Tanner and Wong}{Tanner and
  Wong}{1987}]{tanner87}
Tanner, M. and W.~Wong (1987).
\newblock The calculation of posterior distributions by data augmentation.
\newblock {\em Journal of the American statistical Association\/}~{\em
  82\/}(398), 528--540.

\bibitem[\protect\citeauthoryear{Thijs, Marchal, Lescot, Rombauts, De~Moor,
  Rouze, and Moreau}{Thijs et~al.}{2002}]{thijs02}
Thijs, G., K.~Marchal, M.~Lescot, S.~Rombauts, B.~De~Moor, P.~Rouze, and
  Y.~Moreau (2002).
\newblock A {Gibbs} sampling method to detect overrepresented motifs in the
  upstream regions of coexpressed genes.
\newblock {\em Journal of Computational Biology\/}~{\em 9\/}(2), 447--464.

\bibitem[\protect\citeauthoryear{Thompson, Newberg, Conlan, McCue, and
  Lawrence}{Thompson et~al.}{2007}]{thompson07}
Thompson, W., L.~Newberg, S.~Conlan, L.~McCue, and C.~Lawrence (2007).
\newblock The {Gibbs} centroid sampler.
\newblock {\em Nucleic acids research\/}~{\em 35\/}(suppl 2), W232--W237.

\bibitem[\protect\citeauthoryear{Tompa, Li, Bailey, Church, De~Moor, Eskin,
  Favorov, Frith, Fu, Kent, et~al.}{Tompa et~al.}{2005}]{tompa05}
Tompa, M., N.~Li, T.~Bailey, G.~Church, B.~De~Moor, E.~Eskin, A.~Favorov,
  M.~Frith, Y.~Fu, W.~Kent, et~al. (2005).
\newblock Assessing computational tools for the discovery of transcription
  factor binding sites.
\newblock {\em Nature biotechnology\/}~{\em 23\/}(1), 137--144.

\bibitem[\protect\citeauthoryear{Xing and Karp}{Xing and Karp}{2004}]{xing04}
Xing, E. and R.~Karp (2004).
\newblock {MotifPrototyper}: a bayesian profile model for motif families.
\newblock {\em Proceedings of the National Academy of Sciences of the United
  States of America\/}~{\em 101\/}(29), 10523.

\end{thebibliography}

\end{document}